%% file: main.tex
\documentclass[journal]{IEEEtran}
\usepackage{cite}
\usepackage{amsmath,amssymb,amsfonts}
\usepackage{algorithm,algorithmic}
\usepackage{graphicx}
\usepackage{hyperref}
\hypersetup{hidelinks=true}
\usepackage{textcomp}
\usepackage{float}
\usepackage{enumitem}
\usepackage{amsthm}
\usepackage{comment}
\usepackage{subcaption}
\usepackage{epstopdf}
\newtheorem{thm}{Theorem}
\newtheorem{lem}[thm]{Lemma}
\newtheorem{cor}[thm]{Corollary}
\newtheorem{prop}[thm]{Proposition}

\newtheorem{defn}{Definition}
\newtheorem{rem}[thm]{Remark}
\newtheorem{problem}[thm]{Problem}
\newlist{casenv}{enumerate}{4}
\setlist[casenv]{leftmargin=*,align=left,widest={iiii}}
\setlist[casenv,1]{label={{\itshape\ \casename} \arabic*.},ref=\arabic*}
\setlist[casenv,2]{label={{\itshape\ \casename} \roman*.},ref=\roman*}
\setlist[casenv,3]{label={{\itshape\ \casename\ \alph*.}},ref=\alph*}
\setlist[casenv,4]{label={{\itshape\ \casename} \arabic*.},ref=\arabic*}
\providecommand{\casename}{Case}

\begin{document}

\title{Planar Herding of Multiple Evaders with a Single Herder}
\author{Rishabh Kumar Singh and Debraj Chakraborty 
\thanks{The authors are with the Department of Electrical Engineering, Indian Institute of Technology Bombay, Mumbai 400076, India (e-mail:\tt{\{rishabhsingh,dc\}@ee.iitb.ac.in}
)}}

\maketitle

\begin{abstract}
A planar herding problem is considered, where a superior pursuer herds a flock of non-cooperative, inferior evaders around a predefined target point. An inverse square law of repulsion is assumed between the pursuer and each evader. Two classes of pursuer trajectories are proposed: (i) a constant angular-velocity spiral, and (ii) a constant angular-velocity circle, both centered around the target point. For the spiraling pursuer, the radial velocity is dynamically adjusted based on a feedback law that depends on the instantaneous position of the evader, which is located at the farthest distance from the target at the start of the game. It is shown that, under suitable choices of the model parameters, all the evaders are herded into an arbitrarily small limit cycle around the target point.
Meanwhile, the pursuer also converges onto a circular trajectory around the target. The conditions for the stability of these limit cycles are derived. For the circling pursuer, similar guarantees are provided along with explicit formulas for the radii of the limit cycles.
\end{abstract}


\section{Introduction}
\label{sec:introduction}
Differential games involving multiple evaders and a single pursuer have important applications in various scientific fields, such as robotics \cite{vaughan2000experiments}, behavioral science \cite{wood2007evolving}, wildlife management \cite{bbc_article}, search and rescue operations \cite{van2023steering},  crowd control \cite{hughes2003flow}, and military strategies \cite{chipade2021aerial}. In many cases, the evaders act independently and try to avoid being grouped together, while they tend to move away from the pursuer. On the other hand, the pursuer is quicker than the evaders and aims to gather them into a designated area.

Although there have been numerous attempts in the literature to understand, simulate, and analyze these games and strategies (e.g., see \cite{long2020comprehensive} and the references therein), no formal solutions have been established yet. In this paper, we present simple and effective strategies for the pursuer that ensure the successful herding of any number of evaders, based on an inverse square law of repulsion. The initial exploration of the herding problem was inspired by natural herding behaviors, such as sheepdogs guiding flocks of sheep  \cite{vaughan2000experiments,lien2004shepherding,miki2006effective,strombom2014solving} and predator-prey interactions \cite{hamilton1971geometry,werner1993evolution,scott2013pursuit}. Most early research focused on mathematical models to simulate herding behavior, but paid little attention to theoretical analysis. This lack of theoretical groundwork has made it challenging to develop a comprehensive understanding of the dynamics involved in herding multiple evaders simultaneously.
In contrast, the problem of a single pursuer successfully herding one evader has been effectively addressed through various methods \cite{shedied2002optimal,kachroo2001dynamic,khalafi2011capture}. Similarly, the scenarios of herding multiple evaders with multiple pursuers have also been modeled and simulated \cite{lien2005shepherding,lu2010cooperative,pierson2018controlling,gadre2001learning}. An arc-based method was used in \cite{pierson2018controlling} to herd multiple evaders with multiple pursuers. In \cite{bacon2012swarm}, a sliding mode controller was proposed to guide a single evader along a desired path using a group of pursuers.  

The more complex issue of herding multiple evaders with a single pursuer was explored in \cite{licitra2017singleadaptive,licitra2019single,licitra2017singleswitched}, where different types of switched sliding mode controllers were suggested. However, these solutions depend on the important assumption of non-uniform repulsion between the pursuer and the "chased" and "unchased" evaders. The slower response of the unchased evader allows for the sequential gathering of all evaders at the target using a switched control method. In natural and robotic herding situations, it can be difficult for an individual evader to know if and when it is explicitly being "chased." Typically, most researchers have modeled repulsion based only on the immediate distance from the pursuer \cite{long2020comprehensive}. 

In this paper, we consider the problem of herding multiple evaders with a single pursuer, assuming an inverse square law of repulsion between the pursuer and each evader. The interaction model is similar to that in \cite{pierson2018controlling}. However, unlike \cite{pierson2018controlling}, where multiple pursuers were used, we demonstrate that herding is achievable with just one (sufficiently capable) pursuer. Compared to many of the naturally inspired herding models discussed above, an additional difference in our work is the lack of attraction between the evaders. Although this makes herding more difficult, we believe it suits most engineering applications. We propose two effective herder/pursuer strategies: a constant angular-velocity spiral and a constant angular-velocity circle, both centered around the target point. The spiraling pursuer adjusts its radial velocity based on the instantaneous position of the evader, which was located at the farthest distance from the target at the start of the game. The circular pursuit follows a path with constant angular velocity and radius centered at the target point.  

We show the following.
\begin{enumerate}
    \item For both pursuit strategies, if the evaders are initially located within a specific distance from the target, they all converge onto a limiting circle around the target point (say, of radius $r^{\star}$). For the spiraling strategy, the pursuer simultaneously converges onto another circle (say, of radius $R^{\star}$).
    \item Under suitable assumptions, $r^{\star}$ and $R^{\star}$ are both unique and asymptotically stable.
    \item The magnitude of $r^{\star}$ can be controlled through appropriate pursuer parameters.
    \item For circular pursuit, explicit computation of both $r^{\star}$ and $R^{\star}$ is possible. 
    \item Moreover, for circular herding, the evaders' convergence rate to $r^{\star}$ can be controlled through choices of the pursuit parameters.
\end{enumerate} 
Furthermore, we try to numerically estimate the true regions of attraction for the limit cycles in each scenario since the theoretical estimates of such sets turned out to be intractable. 

This work is an extended version of \cite{rishabh2024}, where the circular pursuit law (see Section \ref{sec:special case}) was proposed. This paper introduces the more versatile spiraling control law (Sections \ref{sec:Herding-of-single} and \ref{sec:Herding-of-multiple}) and provides formal proofs of Lemmas 3, 4, and Theorem 5, which were not included in \cite{rishabh2024}. In addition, several novel numerical examples that demonstrate the effectiveness of the spiraling pursuit law are presented in Section \ref{sec:simulation results}.

\section{Preliminaries and Problem Formulation}\label{sec:prelim} 
Consider a pursuer $P$, with coordinates defined by $(x_p(t),y_p(t))\in\mathbb{R}^2$ in a Cartesian $\{x,y\}$-coordinate frame centered at a target point $z$. The $x$-axis passes through the initial position of the pursuer $(x_p(0),y_p(0))$ as shown in Fig. \ref{fig1}. Consider $n$ evaders $e_{i}$ where $i=\{0,1,2,3,\ldots,n-1\}$, who are required to be herded around $z$ by the pursuer. The index $0$ corresponds to the evader farthest from the target point $z$ at time $t=0$. Denote the position of the $i^{th}$ evader in this frame by $(x_{e_{i}}(t) ,y_{e_{i}}(t))\in\mathbb{R}^2$, $i=\{0,1,2,\ldots,n-1\}$. Let $d_{e_{i}p}(t) =\sqrt{(x_{p}(t) -x_{e_i}(t))^2+(y_{p}(t) -y_{e_i}(t))^2}$  be
the distance between the $i^{th}$ evader and the pursuer and $\hat{d}_{e_{i}p} $ be
the unit vector pointing from $P$ to $e_i$. The distance of the $i^{th}$ evader and the pursuer  
from the origin at time $t$ is $r_{i}(t)$ and $r_p(t)$ respectively. The distance between the initial pursuer position $(x_{p}(0),y_{p}(0))$ and the target point $z$ (i.e. the origin of the $\{x,y\}$ frame) is assumed to be $R$. Since $r_0(0)$ will have a special role to play in our solution, we denote it by $\boldsymbol{\kappa}:=r_0(0)$.
\subsection{Evader Kinematics}
We assume that the evader's instantaneous velocity is proportional to the inverse of the square of the distance between the pursuer and the evader and is pointed directly away from the pursuer at each instant. The following equations capture the aforementioned behavior:

\begin{equation}
\left[\begin{array}{c}
\dot{x}_{e_{i}}(t) \\
\dot{y}_{e_{i}}(t) 
\end{array}\right]=\frac{k}{d_{e_{i}p}^{2}(t)}\hat{d}_{e_{i}p}(t), ~\forall\,i=\{0,1,2,\ldots,n-1\}.\label{eq:evader_dynamics}
\end{equation}
Here, $k$ is a positive constant representing the (identical) repulsion strength between each evader and the pursuer. It is further assumed that there is no attraction or repulsion between the evaders.
\begin{figure}[tbh]
\begin{centering}
\includegraphics[scale=0.8]{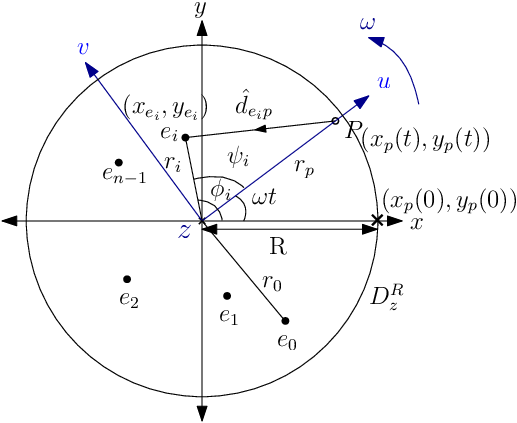}
\par\end{centering}
\centering{}\caption{Evader and pursuer position in Cartesian coordinate system}\label{fig1}
\end{figure}
\subsection{Problem Formulation}\label{subsec:Problem-formulation}
We pose the following problem.
\begin{problem}\label{first_problem}
Find a strategy for the pursuer to herd the evaders to an arbitrarily small neighborhood around the origin (i.e, target $z$). In other words, the pursuer strategy should ensure, for any $\epsilon>0$,
\begin{equation}
\lim_{t\rightarrow\infty}\Vert r_{i}(t) \Vert_2 \le \epsilon,~\forall\,i=\{0,1,2,\ldots,n-1\}.\label{eq:problem_formulation}
\end{equation}
\end{problem}
\textbf{Proposed Strategy:} Let $k_1>0$ and $\omega>0$. We then propose the following strategy for the pursuer:
\begin{equation}
\begin{aligned}
x_{p}(t)  & =R e^{(k_{1}(r_{0}(t) -\boldsymbol{\kappa}))}\cos\left(\omega t\right),\\
y_{p}(t)  & =R e^{(k_{1}(r_{0}(t) -\boldsymbol{\kappa}))}\sin\left(\omega t\right).
\label{eq:pursuer_dynamics}
\end{aligned}
\end{equation}
Hereafter, we suppress the time index $t$ in all the equations for brevity. 

\subsection{Coordinate Frames}\label{sec:coordiante frames}
Apart from the static coordinate frame $\{x,y\}$, we require the following additional frames to describe and simplify the motions of the pursuer and the evading herd:
\begin{enumerate}
    \item Polar coordinates for $\{x,y\}$ denoted by $\{r,\phi\}$:  $x=r\cos\left(\phi\right),y=r\sin\left(\phi\right)$.
\item A rotating frame $\{u,v\}$ with its origin coincident with the target position $z$, with its $u$-axis passing through the pursuer position (Fig.\ref{fig1}) at each instant of time. Hence, this frame rotates counterclockwise with a constant angular velocity of $\omega$.
\item A polar version of the  $\{u,v\}$ frame denoted by $\{r,\psi\}$:  $u=r\cos\left(\psi\right)$, $v=r\sin\left(\psi\right)$. Clearly $\psi=\phi-\omega t\label{psi}$.
\end{enumerate}

\subsubsection{Evader Kinematics in $\{r,\phi\}$-frame}
 By substituting \eqref{eq:pursuer_dynamics} into \eqref{eq:evader_dynamics} and converting the equations into polar coordinates, the coordinates of the $i$-th evader, $\{r_{i}=\sqrt{x^2_{e_i}+y^2_{e_i}},\phi_{i}=\tan^{-1}\frac{y_{e_i}}{x_{e_i}}\}$, $\forall i \in \{0,1,2,\dots,n-1\}$ satisfy:
 {\footnotesize
\begin{equation}
\begin{aligned}\dot{r}_i & =\frac{k\left(r_i-Re^{k_{1}(r_0-\boldsymbol{\kappa})}\cos\left(\phi_i-\omega t\right)\right)}{\left(r_{i}^{2}+R^{2}e^{2k_{1}(r_0-\boldsymbol{\kappa})}-2r_{i}Re^{k_{1}(r_0-\boldsymbol{\kappa})}\cos\left(\phi_i-\omega t\right)\right)^{\frac{3}{2}}},\\
\dot{\phi}_i & =\frac{kRe^{k_{1}(r_0-\boldsymbol{\kappa})}\sin\left(\phi_i-\omega t\right)}{r_i \left(r_{i}^{2}+R^{2}e^{2k_{1}(r_0-\boldsymbol{\kappa})}-2r_iRe^{k_{1}(r_0-\boldsymbol{\kappa})}\cos\left(\phi_{i}-\omega t\right)\right)^{\frac{3}{2}}}.
\end{aligned}
\label{eq:actual_system}
\end{equation}
}%
Clearly, these equations represent a time-varying system. To simplify the analysis, we change coordinates to obtain the following.
\subsubsection{Evader Kinematics in $\{r, \psi\}$-frame}
Define $\psi_i:=\phi_i-\omega t$,\quad$\forall i \in \{0,1,2,\dots,n-1\}$ and substitute this definition into \eqref{eq:actual_system} to obtain the following time-invariant form:
{\footnotesize
\begin{equation}
\begin{aligned}\dot{r}_i & =\frac{k\left(r_i-Re^{k_{1}(r_0-\boldsymbol{\kappa})}\cos\left(\psi_i\right)\right)}{\left(r_{i}^{2}+R^{2}e^{2k_{1}(r_0-\boldsymbol{\kappa})}-2r_iRe^{k_{1}(r_0-\boldsymbol{\kappa})}\cos\left(\psi_i\right)\right)^{\frac{3}{2}}},\\
\dot{\psi}_i & =\frac{kRe^{k_{1}(r_0-\boldsymbol{\kappa})}\sin\left(\psi_i\right)}{r_i\left(r_{i}^{2}+R^{2}e^{2k_{1}(r_0-\boldsymbol{\kappa})}-2r_{i}Re^{k_{1}(r_0-\boldsymbol{\kappa})}\cos\left(\psi_{i}\right)\right)^{\frac{3}{2}}}-\omega.
\end{aligned}
\label{eq:transferred_dynamics}
\end{equation}
}%
\subsubsection{Evader Kinematics in $\{u, v\}$-frame}
Although \eqref{eq:transferred_dynamics} is convenient for analysis, polar coordinates can be less intuitive for visualization. Therefore, we derive the equivalent system of \eqref{eq:transferred_dynamics} in the Cartesian $\{u,v\}$-frame, $\forall i \in \{0,1,2,\dots,n-1\}$:
\begin{equation}
\begin{aligned}
\dot{u}_{i}&=\frac{k\left(u_i-Re^{k_{1}(\sqrt{u_0^2+v_0^2}-\boldsymbol{\kappa})}\right)}{\left((u_i-Re^{k_{1}(\sqrt{u_0^2+v_0^2}-\boldsymbol{\kappa})})^{2}+v_{i}^2\right)^\frac{3}{2}}+\omega v_i,\\
\dot{v}_i&=\frac{kv_i}{\left((u_i-Re^{k_{1}(\sqrt{u_0^2+v_0^2}-\boldsymbol{\kappa})})^{2}+v_{i}^{2}\right)^\frac{3}{2}}-\omega u_i. 
\end{aligned}\label{kinematic_uv}
\end{equation}
\section{Herding of Single Evader}\label{sec:Herding-of-single}
In this section, we examine the behavior of just one evader (i.e. \( i=0 \)) under the action of the pursuer strategy \eqref{eq:pursuer_dynamics}. The analysis of the other evaders, specifically \( i=1,2,\dots,n-1 \), will be presented in the subsequent sections. First note from \eqref{eq:actual_system}, \eqref{eq:transferred_dynamics}, and \eqref{kinematic_uv}, that the motion of \( i=0 \) is independent of \( i=\{1,2,\dots,n-1\} \). 
\subsection{Equilibrium Points}\label{sec:equilibrium_points}
The equilibrium points (say $(r_0^{\star},\psi_0^{\star})$) of \eqref{eq:transferred_dynamics} can be calculated by setting $\dot{r}_0=0$ and $\dot{\psi}_0=0$:
\begin{align}
r_{0}^{\star}&=Re^{k_{1}(r_{0}^{\star}-\boldsymbol{\kappa})}\cos\left(\psi_{0}^{\star}\right)\label{eq:r_eql},\\
\omega &=\frac{kRe^{k_{1}(r_{0}^{\star}-\boldsymbol{\kappa})}\sin\left(\psi_{0}^{\star}\right)}{r_{0}^{\star}\left[(r_{0}^{\star})^{2}+R^{2}e^{2k_{1}(r_{0}^{\star}-\boldsymbol{\kappa})}-2r_{0}^{\star}Re^{k_{1}(r_{0}^{\star}-\boldsymbol{\kappa})}\cos\left(\psi_{0}^{\star}\right)\right]^{\frac{3}{2}}}.\nonumber
\end{align}
Rearranging the above equations, we get 
\begin{equation}
\cos\left(\psi_{0}^{\star}\right)\sin^{2}\left(\psi_{0}^{\star}\right)=\frac{k}{\omega R^{3}e^{3k_{1}(r_{0}^{\star}-\boldsymbol{\kappa})}}.\label{condition_for_eql}
\end{equation}
\subsubsection{Existence of Equilibrium Points}
By eliminating $\psi_{0}^{\star}$ from \eqref{eq:r_eql} and \eqref{condition_for_eql}, we get
\begin{equation}
r_{0}^{\star3}-R^{2}e^{2k_{1}(r_{0}^{\star}-\boldsymbol{\kappa})}r_{0}^{\star}+\frac{k}{\omega}=0.\label{eq:cubic_equation}
\end{equation}
The following lemma provides the condition for the existence of roots of \eqref{eq:cubic_equation}.
\begin{lem}\label{lem:one root existance}
   For the system described in \eqref{eq:transferred_dynamics}, if $\boldsymbol{\kappa}>0$, there exist either one or three positive real roots of \eqref{eq:cubic_equation}, $\forall k,k_1,R,\omega>0$. Furthermore, if \( 2k_{1}^{2}R^{2} > 1 \) and \( \boldsymbol{\kappa}< \frac{\ln(2k_{1}^{2}R^{2})}{2k_{1}} \), then there exists exactly one positive real root of \eqref{eq:cubic_equation}, \( \forall k,w > 0 \).
\end{lem}
\begin{proof}
    Expanding \eqref{eq:cubic_equation}, we get:
      \begin{equation}\label{lem2:c1}
          r_{0}^{*3}-R^{2}e^{-2k_{1}\boldsymbol{\kappa}}( r_{0}^{\star}+2k_{1}r_{0}^{*2}+2k_1^{2}r_{0}^{*3}+\frac{8}{6}k_{1}^{3}r_{0}^{*4}+\dots)+\frac{k}{\omega}=0
      \end{equation}
     If the coefficient of $r_{0}^{*3}$, i.e., $1- 2k_{1}^{2}R^{2}e^{-2k_{1}\boldsymbol{\kappa}}<0$, then the total number of sign changes in the coefficients of \eqref{lem2:c1} is equal to one. By applying an extension of Descartes' rule of signs \cite{curtiss1918recent}, there is only one positive real root. However since $k, R, \boldsymbol{\kappa}>0$, $1- 2k_{1}^{2}R^{2}e^{-2k_{1}\boldsymbol{\kappa}}<0$ can only be satisfied if $2k_{1}^{2}R^{2} > 1$ and $\boldsymbol{\kappa}< \frac{\ln(2k_{1}^{2}R^{2})}{2k_{1}}$. On the other hand, if $1- 2k_{1}^{2}R^{2}e^{-2k_{1}\boldsymbol{\kappa}}>0$, then the total number of sign changes in the coefficients of \eqref{lem2:c1} is equal to three. Again, by applying \cite{curtiss1918recent}, the number of positive roots is either one or three.
\end{proof}
The following lemma computes the equilibrium radius of the pursuer directly from \eqref{eq:pursuer_dynamics}.
\begin{lem}\label{lem:pursuer_radius}
    If the pursuer follows the strategy described by \eqref{eq:pursuer_dynamics} and $r_{0}^{\star}$ is the equilibrium radius for the evader obtained as a solution of \eqref{eq:cubic_equation}, then the pursuer converges to a circular trajectory with a radius given by
    \begin{align}\label{eq:pursuer radius}
        R^{\star} = Re^{k_1(r_{0}^{\star}-\boldsymbol{\kappa})}.
    \end{align}
\end{lem}

From  \eqref{eq:r_eql} and \eqref{eq:pursuer radius}, it is easy to see that 
\begin{equation}\label{eq:r0<R_star}
    r_0^{\star}<R^{\star}.
\end{equation}

\subsubsection{Uniqueness of Equilibrium Points}
The following lemma establishes the uniqueness of the equilibrium point ($r_0^{\star},\psi_0^{\star}$).
\begin{thm}\label{lem:unieq eql point}
    Consider the system described in \eqref{eq:transferred_dynamics}. If $2k_{1}^{2}R^{2}>1$ and $\boldsymbol{\kappa}$ is fixed such that $\boldsymbol{\kappa}<\frac{\ln(2k_{1}^{2}R^{2})}{2k_{1}}$, then there exists a unique equilibrium point $(r_{0}^{\star},\psi_{0}^{\star})$ that satisfies \eqref{eq:r_eql} and \eqref{condition_for_eql}, $\forall k,\omega>0$.
\end{thm}
\begin{proof}
    Recall $r_0^{\star}$ is unique under the assumption $2k_{1}^{2}R^{2}>1$ and $\boldsymbol{\kappa}<\frac{\ln(2k_{1}^{2}R^{2})}{2k_{1}}$. Now, from \eqref{eq:r_eql}, we have:
    \begin{align}\label{eq:cos_psi_star}
        \cos(\psi_{0}^{\star})&=\frac{r_{0}^{\star}}{Re^{k_1(r_{0}^{\star}-\boldsymbol{\kappa})}}\nonumber\\
        & = \frac{1}{Re^{-\boldsymbol{\kappa}k_1}}\frac{r_{0}^{\star}}{e^{k_1r_{0}^{\star}}}.
    \end{align}
    Since all the terms on the right hand side of \eqref{eq:cos_psi_star} are positive, it follows that $\cos(\psi_{0}^{\star})>0$. Now, let $g=\frac{r_{0}^{\star}}{e^{k_1r_{0}^{\star}}}$ then 
    \begin{align*}
        \frac{dg}{dr_{0}^{\star}}=\frac{1-k_1r_{0}^{\star}}{e^{k_1r_{0}^{\star}}}=0\implies r_{0}^{\star}=\frac{1}{k_1}. 
    \end{align*}
    So the maximum value of $g$ is $\frac{1}{k_1e}$. This proves that for arbitrary values of $r_0^{\star}$, $g<\frac{1}{k_1 e}$. Hence
    \begin{align}\label{eq:cos_psi_bound}
        \cos(\psi_{0}^{\star})<\frac{1}{k_1Re^{(1-\boldsymbol{\kappa}k_1)}}.
    \end{align} 
    Now, by hypothesis:
    \begin{equation}\label{eq:a2}
        \boldsymbol{\kappa}<\frac{\ln(2k_{1}^{2}R^{2})}{2k_{1}} \implies - e^{-\boldsymbol{\kappa}k_1}>\frac{1}{\sqrt{2}k_1R}
    \end{equation}
    From \eqref{eq:cos_psi_bound} and \eqref{eq:a2}, we obtain:
    \begin{equation*}
        0<\cos(\psi_0^{\star})<\frac{1}{k_1Re^{(1-\boldsymbol{\kappa}k_1)}}<\frac{\sqrt{2}}{e}<1
    \end{equation*}
Since the value of the cosine function is between \(0\) and \(1\), $\psi_0^{\star}$ will exist for $\frac{-\pi}{2}\leq\psi_0^{\star}\leq\frac{\pi}{2}$. However, two values of \(\psi_{0}^{\star}\) are possible: positive and negative. 

\textbf{Claim:} \(\psi_{0}^{\star}\in (0,\frac{\pi}{2})\)\\ 
Let the equilibrium point in the \(\{u,v\}\)-frame be given as \( u_{0}^{\star} = r_{0}^{\star}\cos(\psi_{0}^{\star}) \) and \( v_{0}^{\star} = r_{0}^{\star}\sin(\psi_{0}^{\star}) \). In the context of \eqref{kinematic_uv}, when we set \( \dot{u}_{0} = 0 \), the equilibrium point satisfies:
\[
\frac{k\left(u_0^{\star}-Re^{k_{1}(\sqrt{u_0^{\star2}+v_0^{\star2}}-\boldsymbol{\kappa})}\right)}{\left((u_0^{\star}-Re^{k_{1}(\sqrt{u_0^{\star2}+v_0^{\star2}}-\boldsymbol{\kappa})})^{2}+v_{0}^{\star2}\right)^{\frac{3}{2}}} = -\omega v_0^{\star}
\]
From \eqref{eq:r0<R_star} and Lemma \( \ref{lem:pursuer_radius} \), we find that the left-hand side of the equation is always negative. This implies that \( v_{0}^{\star} \) must be positive, which in turn indicates that \( \psi_{0}^{\star} \) is always positive.
Thus, we conclude that \( \psi_{0}^{\star} \) is unique.
\end{proof}
{\begin{rem}
   If the distance of the initial position of the pursuer from the target point, i.e., $R$, decreases, then \( k_1 \) must increase to satisfy the inequality \( 2k_1^2 R^2 > 1 \), thereby ensuring the existence of the unique equilibrium point \( (r_{0}^{\star}, \psi_{0}^{\star}) \).
\end{rem}
}
\subsubsection{Stability of the equilibrium point}\label{sec:Stability of equilibrium points}
\begin{thm}\label{thm:For-the-single}
Consider the system described in \eqref{eq:transferred_dynamics}. If  \( 2k_{1}^{2}R^{2}>1 \) and $\boldsymbol{\kappa}$ is fixed such that \( \boldsymbol{\kappa}<\frac{\ln(2k_{1}^{2}R^{2})}{2k_{1}} \), then the unique equilibrium point \( (r_{0}^{\star},\psi_{0}^{\star}) \) is asymptotically stable.
\end{thm}
\begin{proof}
The proof is provided in Appendix A.
\end{proof}

\subsection{Estimation of Stable Region}\label{section3c}
Assume $k,\omega>0$ and $k_1,R$ are fixed such that \( 2k_{1}^{2}R^{2}>1 \). Then we have established that as long as $
\boldsymbol{\kappa} < \frac{\ln(2k_1^2 R^2)}{2k_1}$, there exists a unique asymptotically stable equilibrium point \( (r_0^{\star}, \psi_0^{\star}) \), which depends on \( \boldsymbol{\kappa} =r_0(0)\). However, note that \( (r_0^{\star}, \psi_0^{\star}) \) does not depend on \( \psi_0(0) \). Denote the equilibrium point by $(r_{0}^{\star}(\boldsymbol{\kappa}),\psi_{0}^{\star}(\boldsymbol{\kappa}))$, and the set of all such asymptotically stable equilibrium points as\\
\begin{equation}\label{eq:R_E}
    R_{E}:=\biggl\{(r_{0}^{\star}(\boldsymbol{\kappa}),\psi_{0}^{\star}(\boldsymbol{\kappa}))\in \mathbb{R}^2|\boldsymbol{\kappa}< \frac{\ln(2k_1^2 R^2)}{2k_1}\biggl\}.
\end{equation}
In this section, our objective is to numerically estimate a stable region defined as:
\begin{equation}\label{stable region}
    S = \{(\boldsymbol{\kappa},\psi_{0}(0)) | (r_{0}(t),\psi_{0}(t)) \to R_{E}~~\text{as}~ t\to \infty \}.
\end{equation}
We first denote $(r_{0}^{\star}(\boldsymbol{\kappa}),\psi_{0}^{\star}(\boldsymbol{\kappa}))$ in the Cartesian $\{u,v\}$-frame as $u_{0}^{\star}(\boldsymbol{\kappa})=r_{0}^{\star}(\boldsymbol{\kappa})\cos (\psi_{0}^{\star}(\boldsymbol{\kappa}))$ and $v_{0}^{\star}(\boldsymbol{\kappa})=r_{0}^{\star}(\boldsymbol{\kappa})\sin (\psi_{0}^{\star}(\boldsymbol{\kappa}))$. The equations of motion in this frame were already derived in (\ref{kinematic_uv}). For computational purposes, we further shift the origin of the $\{u,v\}$-frame to $\{u_{0}^{\star}(\boldsymbol{\kappa}), v_{0}^{\star}(\boldsymbol{\kappa})\}$. Denote this new frame as $ \{\bar{u}_0=u_0-u_{0}^{\star}(\boldsymbol{\kappa}),\bar{v}_0=v_0-v_{0}^{\star}(\boldsymbol{\kappa})\}$. Define $\bar{r}_0:=\sqrt{({\bar{u}} + u_{0}^{\star}(\boldsymbol{\kappa}))^2 + ({\bar{v}} + v_{0}^{\star}(\boldsymbol{\kappa}))^2}$, and $h:=((\bar{u}_0 - R e^{k_1(\bar{r}_0 - \boldsymbol{\kappa})} + u_{0}^{\star}(\boldsymbol{\kappa}))^2 + (\bar{v}_0 + v_{0}^{\star}(\boldsymbol{\kappa}))^2)^{3/2}$, then the equations of motion are modified to:
\begin{align}\label{eq:spiral_ROA_shifted_equation_farthest}
    \begin{split}
       \dot{\bar{u}}_0 &= \frac{k(\bar{u}_0 - R e^{k_1(\bar{r}_0 - \boldsymbol{\kappa})} + u_{0}^{\star}(\boldsymbol{\kappa}))}{h} + \omega(\bar{v}_0 + v_{0}^{\star}(\boldsymbol{\kappa})), \\
       \dot{\bar{v}}_0 &= \frac{k(\bar{v}_0 + v_{0}^{\star}(\boldsymbol{\kappa}))}{h} - \omega(\bar{u}_0 + u_{0}^{\star}(\boldsymbol{\kappa})).
    \end{split}
\end{align}

Define \( w_0 = [\bar{u}_0, \bar{v}_0]^T \) and denote \eqref{eq:spiral_ROA_shifted_equation_farthest} as $\dot{w}_0 = f(\boldsymbol{\kappa},w_0)$. For a fixed \( \boldsymbol{\kappa} \), define the function $V(w_0) = w_0^T A_0 w_0, A_0 = A_0^T > 0$. We further define:
\[
\mathcal{A}_0 = \{ w_0 \in \mathbb{R}^2 : V(w_0) \leq 1 \}
\]
and denote its boundary by $\partial \mathcal{A}_0 = \{ w_0 \in \mathbb{R}^2 : V(w_0) = 1 \}$. To estimate $S$, we pose the following problem:
\begin{problem}
    Choose $A_0$ such that the area of $\mathcal{A}_0$ is maximized and 
  \begin{align}\label{eq:lyapunoc_check}
f^T (\boldsymbol{\kappa}, w_0) \frac{\partial V}{\partial w} < 0, \quad \forall w_0 \in \partial \mathcal{A}_0.
\end{align}
\end{problem}
The existence of a solution is guaranteed by the asymptotic stability of the origin. Denote the solution of this problem by $A_0^{\star}$ and define $\bar{\Omega}_0(\boldsymbol{\kappa}) = \{ w_0 \in \mathbb{R}^2 : \omega_0^T A_0^{\star} \omega_0 \leq 1 \}$
We recall the following standard result:
\begin{prop}\cite{hirsch1974differential}\label{proposition}
    A nonempty compact set that is positively or negatively invariant contains either a limit cycle or an equilibrium point.
\end{prop}
Define the set of possible initial conditions for fixed $\boldsymbol{\kappa}$ as $\bar{\mathcal{R}}_0(\boldsymbol{\kappa})=\{(\bar{u}_0(0),\bar{v}_0(0))|\bar{u}_0(0)=\boldsymbol{\kappa}\cos(\psi_0(0))-r_0^{\star}(\boldsymbol{\kappa})\cos(\psi_0^{\star}(\boldsymbol{\kappa}))~ \text{and}~ \bar{v}_0(0)=\boldsymbol{\kappa}\sin(\psi_0(0))-r_0^{\star}(\boldsymbol{\kappa})\sin(\psi_0^{\star}(\boldsymbol{\kappa})), \psi_0(0)\in(\frac{-\pi}{2},\frac{\pi}{2})\}$ and define $\bar{\mathcal{R}}(\boldsymbol{\kappa})=\bar{\mathcal{R}}_0(\boldsymbol{\kappa})\cap\bar{\Omega}_0(\boldsymbol{\kappa})$. Then the following result holds.
\begin{lem}
    For fixed $\boldsymbol{\kappa}<\frac{\ln(2k_1^2R^2)}{2k_1}$, any trajectory of \eqref{eq:spiral_ROA_shifted_equation_farthest} starting from $\bar{\mathcal{R}}(\boldsymbol{\kappa})$, converges asymptotically to the origin.
\end{lem}
\begin{proof}
From \eqref{eq:lyapunoc_check}, the set \( \bar{\Omega}_0(\boldsymbol{\kappa}) \) is a positively invariant set. From Theorem 6, a unique asymptotically stable equilibrium point exists. Then, from Proposition \ref{proposition}, the proof follows.
\end{proof} 
Clearly, the area of $\bar{\Omega}_0(\boldsymbol{\kappa})$, contained in the closed curve \( w_0^T A_0 w_0 = 1 \), is proportional to $\frac{1}{\det(A_0)}$. Inspired by \cite{davison1971computational}, we now pose the following optimization problem for the maximal \( A_0 \).
\begin{problem}\label{optimation_problem}
    $\underset{A_{0} \in \mathbb{R}^{2\times2}}{\text{min}}  \underset{i=1}{\overset{2}{\prod}}\lambda_i(A_{0})~~\text{such that}$
 \begin{enumerate}
    \item $\lambda_{i}(A_{0})>0,~\forall i=1,2$
    \item $f^{T}(\boldsymbol{\kappa}, w_0)A_{0}w_0<0,\forall w_0\in\partial\mathcal{A}$
\end{enumerate}
\end{problem}
Unlike in \cite{davison1971computational}, we use a modern constrained optimization routine (e.g., fmincon in MATLAB \cite{fmincon_matlab}) to obtain solutions to Problem \ref{optimation_problem}. The initial guess for the $A_{0}$ matrix is taken to be the normalized solution of the Lyapunov equation $A_{0}J+J^{T}A_{0}=-\mathbf{I}$, where $J$ is the Jacobian of \eqref{eq:spiral_ROA_shifted_equation_farthest} evaluated at the stable equilibrium point $(\bar{u}_0=0,\bar{v}_0=0)$. For notational convenience, let the sets $\bar{\Omega}_{0}(\boldsymbol{\kappa})$ and $\bar{\mathcal{R}}(\boldsymbol{\kappa})$ are shifted back to the $\{u,v\}$ frame, be denoted by $\Omega_{0}(\boldsymbol{\kappa})$ and $\mathcal{R}(\boldsymbol{\kappa})$.\\

We numerically solve the above optimization for different $\boldsymbol{\kappa}$ values, with fixed  $k_1=1, \omega=2, R=2$, and $k=1$. First, we set $\boldsymbol{\kappa}=1$, and $\Omega_{0}(1)$ is shown as hatched region in Fig. \ref{fig:singh22}. Clearly, $\mathcal{R}(\boldsymbol{\kappa})=\{(r,\psi)| r=\boldsymbol{\kappa}\}$. In particular, for $\boldsymbol{\kappa}=1$, the set $\mathcal{R}(1)=\{(r,\psi)| r=1\}$ (shown as the black circle in Fig. \ref{fig:singh22}) is fully contained in $\Omega_{0}(1)$, i.e., $\mathcal{R}(1)\subset\Omega_{0}(1)$. This implies that any trajectory starting from $\boldsymbol{\kappa}=1$ with an arbitrary $\psi_0(0)$ converges to the asymptotically stable equilibrium point $(r^{\star}(1),\psi^{\star}(1))$ as illustrated in Fig. \ref{fig:singh22}. Recall that $\boldsymbol{\kappa}<\frac{\ln(2k_1^2R^2)}{2k_1}$ is required to guarantee asymptotic stability. Now define $c=\frac{\ln(2k_1^2R^2)}{2k_1}$. The experiment is repeated for a fine grid of values of $\boldsymbol{\kappa}\in[0,c]$.  
Thus an estimate of stable region $S$ (for $k_1=1, \omega=2, R=2, k=1$) can be obtained as: 
\begin{align*}
 \hat{S} = \bigcup\limits_{\boldsymbol{\kappa}\in[0,c]} \mathcal{R}(\boldsymbol{\kappa}),   
\end{align*}
which is drawn as a light blue shaded disk in Fig. \ref{fig:singh22}. Similar results can be obtained for other admissible values of the parameters $\{k, k_1, R, \omega\}$ satisfying the conditions $2k_1^2R^2>1$ and $\boldsymbol{\kappa}<\frac{2k_1^2R^2}{2k_1}$. However, they are not included here due to space constraints.
\begin{figure}[H]
    \centering
        \centering
        \includegraphics[scale=0.3]{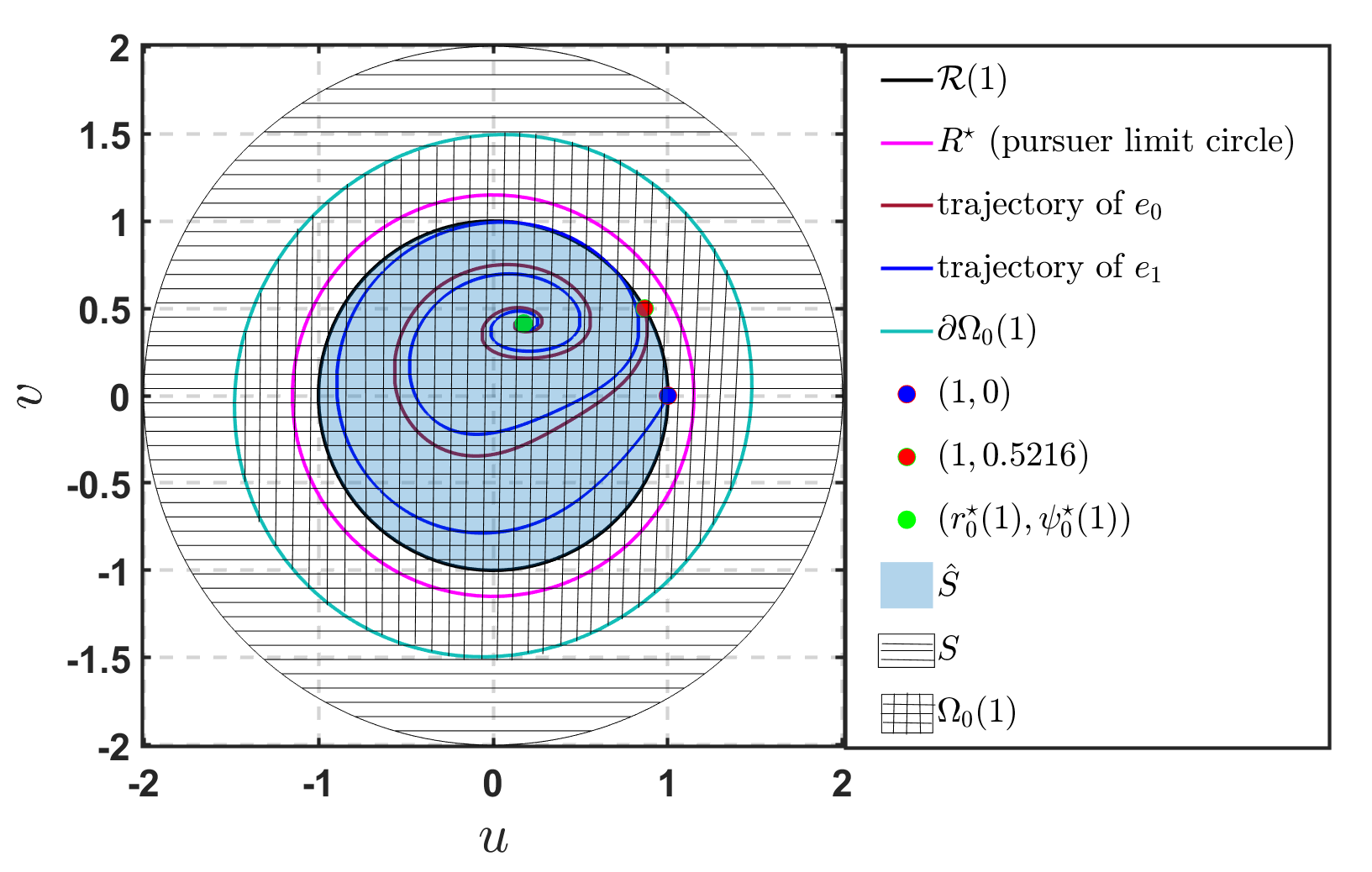}
        \caption{Estimation of $S$ in $(u,v)$-frame for $\boldsymbol{\kappa}=1, R^{\star}=1.15, \omega=2, k=1, k_1=1, R=2$. Two separate trajectories starting from $(1,0)$ and $(1,0.5216)$ converge to $(r_0^{\star}(1),\psi_0^{\star}(1))$.}
        \label{fig:singh22}
\end{figure}

\begin{rem}
    Note that the actual stability region $S$ is significantly larger than $\hat{S}$ shown in Fig. \ref{fig:singh22}. However, due to the multiplicity of possible equilibrium points and the lack of formal stability guarantees, we currently lack an efficient method to rigorously validate the effectiveness of our algorithm beyond $\hat{S}$. Nevertheless, for this specific example, we illustrate the actual region of stability in Fig. \ref{fig:singh22} through brute force computation.
\end{rem}
\subsection{Translation of Results to $\{x, y\}/\{r, \phi\}$-frame}\label{sec:Translation of Results}
The previous results almost entirely dealt with the $\{u,v\}/\{r,\psi\}$ frame, mainly due to the time-invariance of the system equations in this frame. However, we need to interpret all the results in the $\{x,y\}$/$\{r,\phi\}$ frame for physical understanding of the solution. In all the trajectories below, we use the fact that the solutions $r_0(t)$ remain the same in both $\{u,v\}$ and $\{x,y\}$ frames. On the other hand, any $\psi_0$ in the $\{u,v\}$ frame translates to $\phi_0=\psi_0+\omega t$ in the $\{x,y\}$ frame. Hence the point ($r_{0}^{\star},\psi_{0}^{\star}$) corresponds to the circular trajectory defined by radius $r_{0}^{\star}$ and the angle evolving according to $\phi_{0}^{\star}:=\psi_{0}^{\star}+\omega t$.  
    \begin{thm}
        Consider the system described in (\ref{eq:actual_system}). Assume $2k_{1}^{2}R^{2}>1$ and $\boldsymbol{\kappa}$ is fixed such that $\boldsymbol{\kappa}<\frac{\ln(2k_{1}^{2}R^{2})}{2k_{1}}$, then (\ref{eq:actual_system}) has a unique asymptotically stable limit cycle described by the trajectory $L_0:=\{r_{0}^{\star}=r_{0}^{\star}, \phi_{0}^{\star}=\psi_{0}^{\star}+\omega t,~\forall t\geq 0\}$, $\forall k,\omega>0$.
        \label{thm:eql set theorem}
    \end{thm}
    \begin{proof}
    The proof follows directly from Theorem \ref{lem:unieq eql point} and Theorem \ref{thm:For-the-single}. 
    \end{proof}

\section{Herding of Multiple Evader}\label{sec:Herding-of-multiple}
In this section, we collectively analyze the behavior of \eqref{eq:transferred_dynamics}, for $ i=\{0,1,2,\dots,n-1\}$.
\subsection{Equilibrium Points}
To determine the equilibrium points of \eqref{eq:transferred_dynamics}, we set $\dot{r}_i=0$ and $\dot{\psi}_i=0$, $\forall i \in \{0,1,2,\dots,n-1\}$.
Denote $(r_i^{\star}, \psi_i^{\star})$ as the equilibrium points. For each evader $i$, the equilibrium conditions yield:
\begin{subequations}\label{eq:multiple_eql_condition}
    \begin{align}
        r_{i}^{\star}&=Re^{k_{1}(r_{0}^{\star}-\boldsymbol{\kappa})}\cos\left(\psi_{i}^{\star}\right),\label{eq:eql_condition_a}\\
\cos\left(\psi_{i}^{\star}\right)\sin^{2}\left(\psi_{i}^{\star}\right)&=\frac{k}{\omega R^{3}e^{3k_{1}(r_{0}^{\star}-\boldsymbol{\kappa})}}.\label{eq:eql_condition_b}
    \end{align}
\end{subequations}
\begin{thm}\label{lem:rstar}
     If $2k_1^2R^2>1$ and $\boldsymbol{\kappa}$ is fixed such that $\boldsymbol{\kappa}<\frac{\ln(2k_1^2R^2)}{2k_1}$, then \eqref{eq:transferred_dynamics} has a unique equilibrium point defined by $r_i^{\star}=r^{\star}, \psi_i^{\star}=\psi^{\star},  \forall i={0,\dots,n-1}, \forall k,\omega>0$. Moreover, $r^{\star}$ is the unique positive real root of
    \begin{align*}
        r^{\star3}-R^{2}e^{2k_1(r_0^{\star}-\boldsymbol{\kappa})}r^{\star}+\frac{k}{\omega}=0~~~ \text{and}\\
        \psi^{\star}=\cos^{-1}\left(\frac{r^{\star}}{Re^{k_1(r_0^{\star}-\boldsymbol{\kappa})}}\right)>0.
    \end{align*}
\end{thm}
\begin{proof}
    Eliminating $\psi_i^{\star}$ from \eqref{eq:multiple_eql_condition}, we derive:
    \begin{align}\label{thm14:c1}
         r_i^{\star3}-R^{2}e^{2k_1(r_0^{\star}-\boldsymbol{\kappa})}r_i^{\star}+\frac{k}{\omega}=0.
    \end{align}
    Under the assumptions of this theorem, the existence and uniqueness of the root of \eqref{thm14:c1} has been established already for $i=0$ (see Section \ref{sec:Herding-of-single} and Theorem \ref{lem:unieq eql point}). This root was denoted by $r_0^{\star}$.
    Subtracting \eqref{eq:cubic_equation} from \eqref{thm14:c1} gives:
    \begin{align*}
        r_i^{\star3}-r_0^{\star3}+R^2e^{2k_1(r_0^{\star}-\boldsymbol{\kappa})}(r_0^{\star}-r_i^{\star})=0.
    \end{align*}
    Simplifying, we obtain:
    \begin{align}\label{eq:thm14 c1}
        (r_i^{\star}-r_0^{\star})(r_i^{\star2}+r_0^{\star2}+r_i^{\star}r_0^{\star}+R^2e^{2k_1(r_0^{\star}-\boldsymbol{\kappa})})=0.
    \end{align}
    Since $r_i^{\star2}+r_0^{\star2}+r_i^{\star}r_0^{\star}+R^2e^{2k_1(r_0^{\star}-\boldsymbol{\kappa})}>0$, it follows from \eqref{eq:thm14 c1} that 
    \begin{align}\label{eq:common_eql}
    r_i^{\star}=r_0^{\star}=:r^{\star},~ \forall i\in\{1,2,\dots,n-1\}.
    \end{align}
    From \eqref{eq:common_eql}, \eqref{eq:cubic_equation} and Theorem \ref{lem:unieq eql point}, the uniqueness and positivity of  $\psi_{i}^{\star}$ follow, thereby allowing us to write:
     \begin{align*}
    \psi_i^{\star}=\psi_0^{\star}=:\psi^{\star}>0,~ \forall i\in\{0,1,2,\dots,n-1\}.
    \end{align*}
    
\end{proof}
In the following result, we analyze the effect of increasing $\omega$ on the location of the equilibrium points.
\begin{lem}\label{lem:r_limit_multiple}
For system \eqref{eq:transferred_dynamics}, under the assumptions of Theorem \ref{lem:rstar} and the equilibrium radius $r^{\star}$ defined as the solution of \eqref{eq:cubic_equation}, satisfies $r^{\star}\rightarrow 0$ as $\omega\rightarrow\infty$.
\end{lem}
\begin{proof}
From \eqref{eq:cubic_equation}, we have:
\begin{equation*}
r^{\star3}-R^{2}e^{2k_1(r^{\star}-\boldsymbol{\kappa})}r^{\star}+\frac{k}{\omega}=0.    
\end{equation*}
As $\omega\rightarrow\infty$, we have $\frac{k}{\omega}\rightarrow0$
which implies:
\begin{equation*}
r^{\star3}-R^{2}e^{2k_1(r^{\star}-\boldsymbol{\kappa})}r^{\star}=0
\Longleftrightarrow r^{\star}(r^{\star2}-R^{2}e^{2k_1(r^{\star}-\boldsymbol{\kappa})})=0 
\end{equation*}
This implies: either $r^{\star}=0$ or $r^{\star2}-R^{2}e^{2k_1(r^{\star}-\boldsymbol{\kappa})}=0$.\\ Now, let:
\begin{equation*}
  q = r^{\star2}-R^{2}e^{2k_1(r^{\star}-\boldsymbol{\kappa})}.
\end{equation*}
Expanding $q$, we get:
\begin{equation*}
    q = r^{\star2}-R^2e^{-2k_{1}\boldsymbol{\kappa}}(2k_{1}r^{\star}+\frac{(2k_{1}r^{\star})^2}{2!}+\dots).
\end{equation*}
Since $\boldsymbol{\kappa}<\frac{\ln(2k_{1}^{2}R^{2})}{2k_{1}}$ by assumption the coefficient of $r^{\star2}$ is $1-2k_{1}^{2}R^2e^{-2k_{1}\boldsymbol{\kappa}}<0$, which implies that there are no sign changes in the coefficients. Hence, by \cite{curtiss1918recent}, there are no positive real roots for $r^{\star}$ that satisfy $q=0$.
\end{proof}

\subsection{Stability of the equilibrium point}
The asymptotic stability of the unique equilibrium point is established in the following result.
\begin{thm}\label{thm:multiple_spiral}
    Let $2k_{1}^{2}R^{2}>1$ and $\boldsymbol{\kappa}$ is fixed such that $\boldsymbol{\kappa}<\frac{\ln(2k_{1}^{2}R^{2})}{2k_{1}}$ and $\forall k,\omega>0$. Then, the unique equilibrium point of \eqref{eq:transferred_dynamics} defined by $(r^{\star},\psi^{\star})$ (see Theorem \ref{lem:rstar}) is asymptotically stable for all the evaders. 
\end{thm}
\begin{proof}
    The proof is provided in Appendix B.
\end{proof}
\subsection{Estimation of Stable Region}
It is evident that the common equilibrium point defined in Theorem \ref{lem:rstar}, $\forall i={0,\dots,n-1}$, depends on $\boldsymbol{\kappa}$. From Theorem \ref{lem:rstar}, it also follows that the set of equilibrium points for varying $\boldsymbol{\kappa}$ is identical for all evaders. Hence, we can denote the common equilibrium set using the notation $R_{E}$ introduced in \eqref{eq:R_E} earlier. However, the trajectory of the $i^{th}$ evader $(i=1,2,\dots,n-1)$ depends on the starting position of the $i=0$ evader, i.e., $(\boldsymbol{\kappa},\psi_0(0))$, as well as its own initial condition $r_i(0),\psi_i(0)$. Hence, corresponding to $R_E$, we can define a stability region for the $i^{th}$ evader, $i=\{1,2,\dots,n-1\}$ indexed by $(\boldsymbol{\kappa},\psi_0(0))$ as follows:
\begin{align}\label{eq:stability set}
    S^{i}(\boldsymbol{\kappa},\psi_0(0))=\{(r_i(0),\psi_i(0))|(r_i(t),\psi_i(t))\to R_E\}
\end{align}
Since no efficient method of estimation of $S^{i}(\boldsymbol{\kappa},\psi_0(0))$ is known, a brute force computation of this set is shown in Fig. \ref{fig:singh34} for an arbitrary evader $i\in\{1,2,\dots,n-1\}$ for $\boldsymbol{\kappa}=1, \psi_0(0)=0, R=2, k_1=1, k=1$.
\begin{figure}[H]
\hspace{-0.603cm}
    \centering
        \includegraphics[scale=0.248]{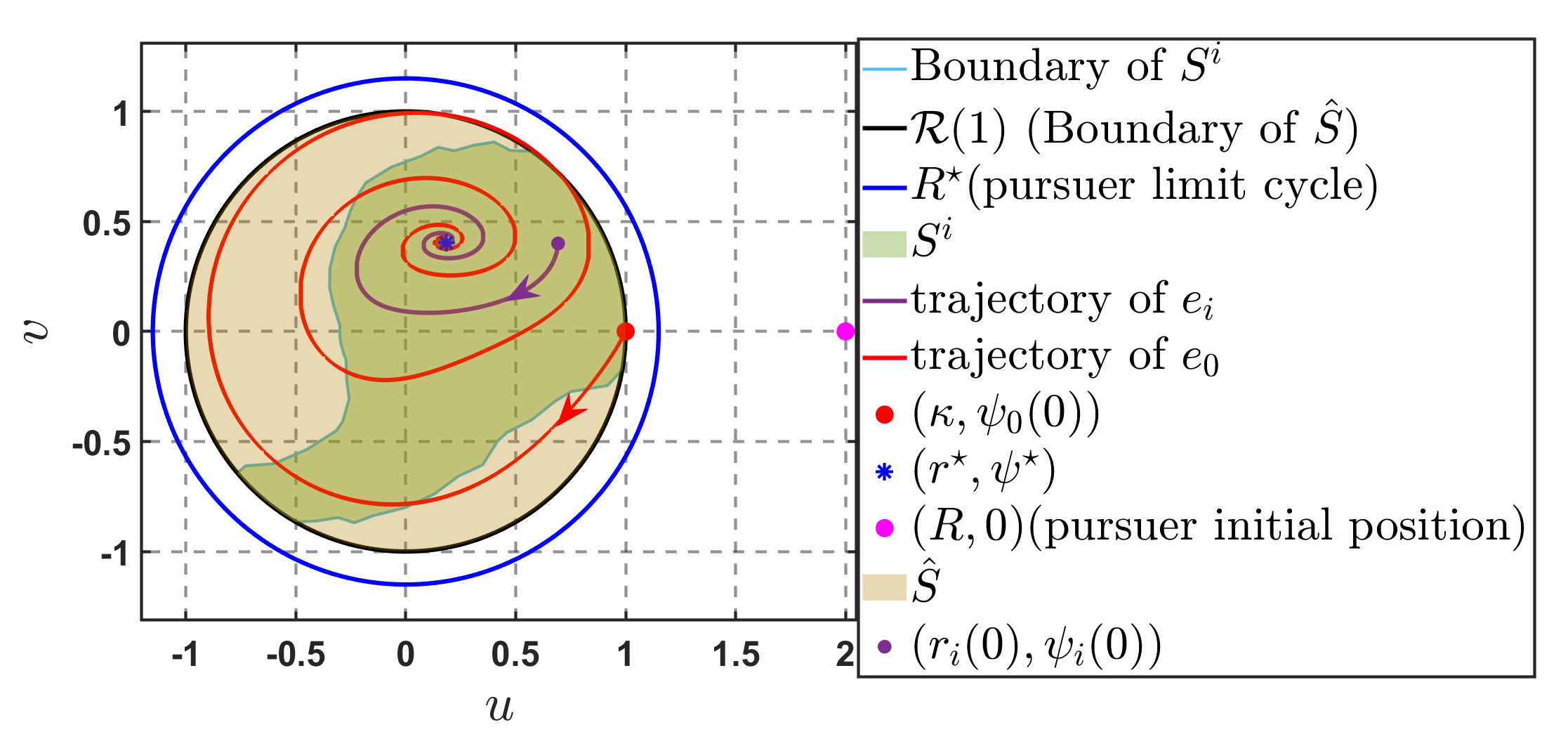}
    \caption{Estimation of $S^{i}$ in $(u,v)$-frame for $\boldsymbol{\kappa}=1, \psi_0(0)=0, R=2, k_1=1, k=1$}
    \label{fig:singh34}
\end{figure}
This figure shows that if the pursuer starts from $(R,0)$ (indicated by a magenta dot), the outermost evader $e_0$ starts from $(\boldsymbol{\kappa},\psi_0(0))$, and if another evader $e_i$ starts at $(r_i(0),\psi_i(0))$ from within $S^{i}$ (the light green shaded region), then the trajectories of both the evaders (red and purple curves, respectively) converge to the equilibrium point $(r^{\star},\psi^{\star})$. Meanwhile, the pursuer stabilizes onto a circular trajectory of radius $R^{\star}$ (shown as a blue circle).
\subsection{Translation of Results to $\{x, y\}/\{r, \phi\}$-frame}\label{sec:Translation of Results multiple}
Similar to Section \ref{sec:Translation of Results}, all the results can be interpreted in the $\{x, y\}/\{r, \phi\}$-frame.
    \begin{thm}
        Consider the system described in \eqref{eq:actual_system}. Let $2k_{1}^{2}R^{2}>1$ and $\boldsymbol{\kappa}$ is fixed such that $\boldsymbol{\kappa}<\frac{\ln(2k_{1}^{2}R^{2})}{2k_{1}}$. Then (\ref{eq:actual_system}) has a unique asymptotically stable limit cycle described by the trajectory $L:=\{r^{\star}=r^{\star}, \phi^{\star}=\psi^{\star}+\omega t,~\forall t\geq 0\}$, $\forall k,\omega>0$.
        \label{thm:eql set theorem multiple}
    \end{thm}
    \begin{proof}
    The proof follows directly from Lemma \ref{lem:rstar} and Theorem \ref{thm:multiple_spiral}. 
    \end{proof}
        The following result shows that the evader can be driven arbitrarily close to the target by choosing a sufficiently large pursuer velocity.
    \begin{cor}
    For any $\epsilon>0, \exists W>0$ such that if $\omega>W$, then the radius of $L$ i.e. $r^{\star}<\epsilon$
    \end{cor}
    \begin{proof}
        The proof follows from Lemma \ref{lem:r_limit_multiple}.
    \end{proof}
    
\section{A Simplified Herding Law}\label{sec:special case}
We consider a special case of the proposed strategy of the pursuer when $k_1=0$. Note that this assumption fundamentally changes the pursuer's behavior since it can no longer adjust the radius of its herding circle. Refer to Sections \ref{subsec:Problem-formulation} and \ref{sec:coordiante frames} for the coordinate frames and evader kinematics while considering $k_1=0$. To distinguish from the equations discussed above, we add the subscript 's' to each symbol: $\{x,y\} \to \{x_s,y_s\}$, $\{r,\phi\} \to \{r_s,\phi_s\}$, $\{u,v\} \to \{u_s,v_s\}$, $\{r,\psi\} \to \{r_s,\psi_s\}$.

\subsection{Computation of Equilibrium Points}\label{sec:special_equilibrium_points}
If we put $k_1=0$ in \eqref{eq:multiple_eql_condition}, we get (denoting ($r^{i\star}_{s},\psi^{i\star}_{s}$) as the equilibrium point for the $i^{th}$ evader):
\begin{subequations}
\begin{align}
r^{i\star}_{s}&=R\cos\left(\psi^{i\star}_{s}\right)\label{eq:special_r_eql}
\end{align} 
\begin{align}
\cos\left(\psi^{i\star}_{s}\right)\sin^{2}\left(\psi^{i\star}_{s}\right)=\frac{k}{\omega R^{3}}.\label{special_condition_for_eql}
\end{align}
\end{subequations}
\subsubsection{Existence of Equilibrium Points}
We know that $-\frac{2}{3\sqrt{3}}< \cos\left(\psi^{i\star}_{s}\right)\sin^{2}\left(\psi^{i\star}_{s}\right)<\frac{2}{3\sqrt{3}}$ (see Fig. \ref{singh15}) and $k,\omega$ and $R$ are all positive. This implies $0<\frac{k}{\omega R^{3}}<\frac{2}{3\sqrt{3}}$.
\begin{figure}[tbh]
\begin{centering}
\includegraphics[scale=0.5]{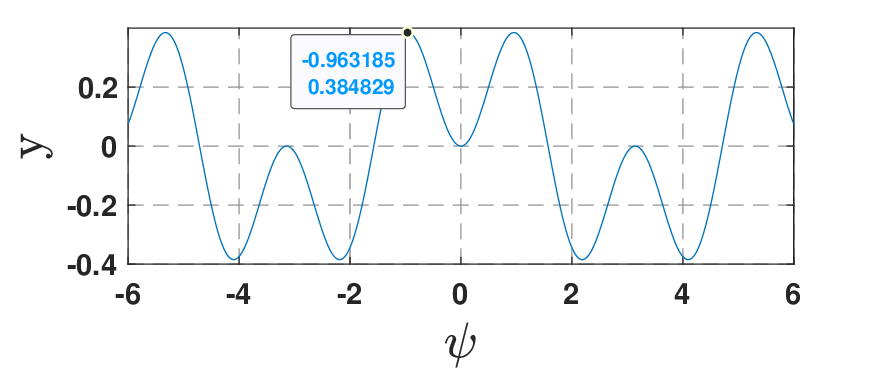}
\par\end{centering}
\centering{}\caption{Graph for $y=\cos(\psi)\sin^{2}(\psi)$}\label{singh15}
\end{figure}
The next Lemma follows immediately.
\begin{lem}\label{lem_existance_of_eql_points}
If $k_1=0$ is substituted in \eqref{eq:transferred_dynamics}, then the resulting system has equilibrium points if and only if $0<\frac{k}{\omega R^{3}}<\frac{2}{3\sqrt{3}}$. Equivalently, system \eqref{eq:actual_system} (with $k_1=0$) has limit sets/cycles if and only if $0<\frac{k}{\omega R^{3}}<\frac{2}{3\sqrt{3}}$.
\end{lem}
As in previous cases, since $\phi^{i}_{s} :=\psi^{i}_{s} +\omega t$, the point ($r^{i\star}_{s},\psi^{i\star}_{s}$) corresponds to the circular trajectory defined by radius $r^{i\star}_{s}$ and the angle evolving according to $\phi^{i\star}_{s} :=\psi^{i\star}_{s}+\omega t$.

\subsubsection{Computation of Equilibrium Points}
Eliminating $\psi^{i\star}_{s}$ from \eqref{eq:special_r_eql} and \eqref{special_condition_for_eql}, we get
\begin{equation}
r^{i*3}_s-R^{2}r^{i\star}_{s}+\frac{k}{\omega}=0 \label{eq:special_cubic_equation}
\end{equation}
Denote the three roots of \eqref{eq:special_cubic_equation} by $r^{i\star}_{s_{1}}$, $r^{i\star}_{s_{2}}$, and $r^{i\star}_{s_{3}}$. Note that $\frac{k}{\omega R^{3}}<\frac{2}{3\sqrt{3}} \implies \frac{-R^{6}}{27}+\frac{k^{2}}{4\omega^{2}}<0$. From the theory of general cubic equations \cite{abramowitz1968handbook}, this implies that all the roots are real. If we denote $\sigma_{1}=\left(-\sqrt{\frac{k^{2}}{4\omega^{2}}-\frac{R^{6}}{27}}-\frac{k}{2\omega}\right)^{\frac{1}{3}}$ and $\sigma_{2}=\left(\sqrt{\frac{k^{2}}{4\omega^{2}}-\frac{R^{6}}{27}}-\frac{k}{2\omega}\right)^{\frac{1}{3}}$, then
\begin{equation}
\begin{aligned}
r^{i\star}_{s_{1}}&=\sigma_2+\sigma_1,\\
r^{i\star}_{s_{2}}&=\frac{-\sigma_{2}}{2}+\frac{\sqrt{3}\sigma_{1}}{2}j-\frac{\sigma_{1}}{2}-\frac{\sqrt{3}\sigma_{2}}{2}j,\\
r^{i\star}_{s_{3}}&=\frac{-\sigma_{2}}{2}-\frac{\sqrt{3}\sigma_{1}}{2}j-\frac{\sigma_{1}}{2}+\frac{\sqrt{3}\sigma_{2}}{2}j.
\end{aligned}\label{special_threeroots}
\end{equation}
\begin{lem}\label{special_lem_threeroots}
For $\frac{k}{\omega R^{3}}<\frac{2}{3\sqrt{3}}$ and for finite $k, \omega, R>0$, the three roots \eqref{special_threeroots} satisfy $R>r^{i\star}_{s_{1}}>\frac{R}{\sqrt{3}}>r^{i\star}_{s_{2}}>0>r^{i\star}_{s_{3}}$, $\forall i\in\{0,1,\dots,n-1\}$.
\end{lem}

\begin{proof}
Using $\frac{k}{\omega R^{3}}<\frac{2}{3\sqrt{3}}$ in \eqref{special_threeroots}, one can show $R>r^{i\star}_{s_{1}}>r^{i\star}_{s_{2}}>0>r^{i\star}_{s_{3}}$. To calculate the intermediate bound of $\frac{R}{\sqrt{3}}$, let $y=r^{i*3}_s-R^{2}r^{i\star}_{s}+\frac{k}{\omega}$ and consider the graph of this cubic equation as shown in Fig. \ref{special_Graph_of_cubic_equation}.
\begin{figure}[tbh]
\begin{centering}
\includegraphics[scale=0.75]{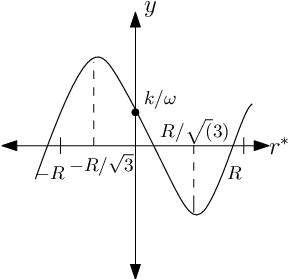}
\par\end{centering}
\caption{Graph of $y=r^{i*3}_s-R^{2}r^{i\star}_{s}+\frac{k}{\omega}$}
\label{special_Graph_of_cubic_equation}
\end{figure}
Notice that the minimum and maximum value of $y$ occur at $r^{i\star}_{s}=\pm\frac{R}{\sqrt{3}}$.
It is easy to verify that, since $\frac{k}{\omega}<\frac{2}{3\sqrt{3}}R^{3}$, $y\left(\frac{R}{\sqrt{3}}\right)=-\frac{2}{3\sqrt{3}}R^{3}+\frac{k}{\omega}<0$. On the other hand,  $y\left(\frac{-R}{\sqrt{3}}\right)=\frac{4R^{3}}{3\sqrt{3}}+\frac{k}{\omega}>0$. This implies that the maximum and the minimum (as drawn in the Fig. \ref{special_Graph_of_cubic_equation}) occur at $\frac{-R}{\sqrt{3}}$ and $\frac{R}{\sqrt{3}}$ respectively. The claim follows directly from the figure.
\end{proof}

Clearly, the negative root $r^{i\star}_{s_{3}}$ is not relevant to our analysis since the radius must be positive. The other roots $r^{i\star}_{s_{1}}$, $r^{i\star}_{s_{2}}$ each correspond to an equilibrium point for \eqref{eq:transferred_dynamics} with $k_1=0$. Using arguments similar to Theorem \ref{lem:unieq eql point}, the corresponding angles $\psi_{s_{1}}^{\star}=\cos^{-1}\frac{r^{i\star}_{s_{1}}}{R}$ and $\psi_{s_{2}}^{\star}=\cos^{-1}\frac{r^{i\star}_{s_{2}}}{R}$ are positive and unique. We observe that the entire analysis is independent of the initial position of each evader $e_{i}$. Thus, we have the following Theorem:
\begin{thm}\label{lem:two_eql_points}
        Consider \eqref{eq:transferred_dynamics} with $k_1=0$. Then, if $0<\frac{k}{\omega R^3}<\frac{2}{3\sqrt{3}}$, each evader $(i=0,\dots,n-1)$ has two equilibrium points $(r_{s_{1}}^{\star},\psi_{s_{1}}^{\star})$ and $(r_{s_{2}}^{\star},\psi_{s_{2}}^{\star})$ satisfying $R>r_{s_{1}}^{\star}>\frac{R}{\sqrt{3}}>r_{s_{2}}^{\star}>0$. Correspondingly,  $0<\psi_{s_{1}}^{\star}=cos^{-1}(\frac{r_{s_{1}}^{\star}}{R})<\frac{\pi}{2}$ and $0<\psi_{s_{2}}^{\star}=cos^{-1}(\frac{r_{s_{2}}^{\star}}{R})<\frac{\pi}{2}$.
\end{thm}
In Fig. \ref{singh17}, these two equilibrium points are depicted in the $\{u,v\}$-frame with the red and blue dots, respectively.
\begin{figure}[tbh]
\begin{centering}
\includegraphics[scale=0.8]{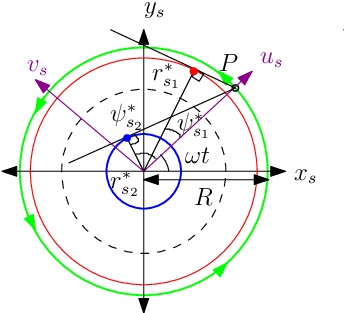}
\par\end{centering}
 \caption{Green-pursuer circle, red-unstable limit cycle, blue-stable limit cycle, dotted-estimation of ROA}
\label{singh17}
\end{figure}
Using $\phi^{\star}_{s} :=\psi^{\star}_{s}+\omega t$, the equilibrium points $(r^{\star}_{s_{i}}, \psi^{\star}_{s_{i}})$, $i=1,2$, correspond to the limit cycles $(r^{\star}_{s_{i}}, \phi^{\star}_{s_{i}}) = (r^{\star}_{s_{i}}, \psi^{\star}_{s_{i}} + \omega t)$ for $\forall t>0$. These limit cycles are shown in Fig. \ref{singh17} by the red and blue circles, respectively.

In the following result, we analyze the effect of increasing $\omega$ on the location of the equilibrium points.

\begin{cor}\label{special_limit_eq}
For system \eqref{eq:transferred_dynamics} with $k_1=0$, the equilibrium radii defined in \eqref{special_threeroots} satisfies
$r^{i\star}_{s_{2}}\rightarrow 0$ and $r^{i\star}_{s_{1}}\rightarrow R$ as $\omega\rightarrow\infty$, $\forall i\in\{0,1,\dots,n-1\}$.
\end{cor}
\begin{proof}
From \eqref{eq:special_cubic_equation}, we have $r^{i*3}_s-R^{2}r^{i\star}_{s}+\frac{k}{\omega}=0$.
If $\omega\rightarrow\infty$, then $\frac{k}{\omega}\rightarrow0$.
This implies $r^{i*3}_s-R^{2}r^{i\star}_{s}=0
\Longleftrightarrow r^{i\star}_{s}(r^{i\star2}_{s}-R^{2})=0$. Hence, either $r^{i\star}_{s}=0$ or $r^{i\star}_{s}=\pm R$. The result follows since
$r^{i\star}_{s}$ is always positive. 
\end{proof}
\begin{cor}\label{col:r_star<epsilon}
For any $\epsilon>0, \exists W>0$ such that if $\omega>W$, then the radius of $L_{2}$ i.e. $r^{\star}_{s_{2}}<\epsilon$
\end{cor}
\subsection{Stability of equilibrium points}\label{sec:special_Stability of equilibrium points}
\begin{thm}\label{thm:stability}
For the system described in \eqref{eq:transferred_dynamics} with $k_1=0$,
let $0<\frac{k}{\omega R^{3}}<\frac{2}{3\sqrt{3}}$, and $(r_{s_{1}}^*,\psi_{s_{1}}^*)$, $(r_{s_{2}}^*,\psi_{s_{2}}^*)$ be the equilibrium points defined in \eqref{special_threeroots} and Theorem \ref{lem:two_eql_points}. Then 
\begin{enumerate}
    \item$(r_{s_{2}}^{\star},\psi_{s_{2}}^{\star})$ is asymptotically stable
    \item $(r_{s_{1}}^*,\psi_{s_{1}}^*)$ is a saddle point
\end{enumerate}
\end{thm}

\begin{proof}
Linearizing \eqref{eq:transferred_dynamics} with $k_1=0$ at any equilibrium point $(r^{\star}_{s},\psi^{\star}_{s})$, the Jacobian turns out to be: 
\begin{align}
J & = \begin{bmatrix}
\frac{k}{\left(R\sin\left(\psi^{\star}_{s}\right)\right)^{3}} & \frac{k}{\left(R\sin\left(\psi^{\star}_{s}\right)\right)^{2}} \\
\frac{-k R^{2}\cos^{2}\left(\psi^{\star}_{s}\right)}{\left(R\sin\left(\psi_{0}^{\star}\right)\right)^{2}} & \frac{-2k}{R^{3}\sin^{3}\left(\psi^{\star}_{s}\right)}
\end{bmatrix}
\end{align}

The eigenvalues of $J$ can be readily computed as:
\begin{align*}
\lambda_{1,2} & =\frac{-k}{2R^{3}\sin^{3}\left(\psi^{\star}_{s}\right)}\pm\frac{1}{2}\sqrt{\frac{9k^{2}}{R^{6}\sin^{6}\left(\psi^{\star}_{s}\right)}-4\omega^{2}}.\\
\end{align*} 
From \eqref{eq:special_r_eql} and \eqref{special_condition_for_eql}, $sin^{3}\left(\psi^{\star}_{s}\right)=\left(\frac{\sqrt{R^{2}-r^{\star2}_{s}}}{R}\right)^{3}$\!.
Hence
\begin{equation}\label{special_eigen_value}
\lambda_{1,2}=\frac{1}{2(R^{2}-r^{\star2}_{s})^{\frac{3}{2}}}\left(-k\pm\sqrt{9k^{2}-4\omega^{2}(R^{2}-r^{\star2}_{s})^{3}}\right)
\end{equation}
Now we know that $r^{\star}_{s_{2}}<\frac{R}{\sqrt{3}}<r^{\star}_{s_{1}}$. We claim that for $r^{\star}_{s}=r^{\star}_{s_{2}}$ the maximum real value of $9k^{2}-4\omega^{2}(R^{2}-r^{\star}_{s_{2}})^{3}<k^2$. Since $r^{\star}_{s_{2}}<\frac{R}{\sqrt{3}}$, we can write 
\begin{equation}
\begin{aligned}
9k^{2}-4\omega^{2}(R^{2}-r^{\star^2}_{s})^{3} &< 9k^{2}-4\omega^{2}(R^{2}-\frac{R^2}{3})^{3}\\
&=9k^{2}-\frac{4\times 8}{27}(\omega^{2} R^6)
\end{aligned}
\end{equation}
However, we have assumed that $\omega>\frac{3\sqrt{3}k}{2R^{3}}$ i.e. $\omega^{2}>\frac{27k^{2}}{4R^{6}}$. Hence $9k^{2}-\frac{4\times 8}{27}(\omega^{2} R^6)<9k^2-8k^2=k^2$. Claim 1 follows. See Appendix C for the proof of claim 2.
\end{proof}
\begin{rem} Note that the real parts of the eigenvalues of the Jacobian $J$ are roughly proportional to the inverse of the cube of the pursuer radius (see \eqref{special_eigen_value}). This implies that the decreasing the radius accelerates the herding action, assuming that one can increase the $\omega$ appropriately to keep satisfying  $0<\frac{k}{\omega R^{3}}<\frac{2}{3\sqrt{3}}$.
\end{rem}

\begin{rem}
As shown in Corollaries \ref{special_limit_eq} and \ref{col:r_star<epsilon}, one can increase $\omega$ to ensure that the stable equilibrium point $r^{\star}_{s_{2}}$ is arbitrarily close to the origin. Consequently, one can independently control the final distance of the evader from the target and the rate of convergence to this final position conveniently through $\omega$ and $R$, respectively.  
\end{rem}

\subsection{Translation of Results to $\{x_s,y_s\}$/$\{r_s,\phi_s\}$-frame}\label{subsubsection:translation of results to static frame}
We can interpret all the results in the $\{x_s,y_s\}$/$\{r_s,\phi_s\}$ frame for physical understanding of the solution. Let us denote the disc centered at $z$ with radius $R$ as $D_z^R:=\{x\in\mathbb{R}^2:\|{x}\|_2\le R\}$.
    \begin{thm}
        Consider system (\ref{eq:actual_system}) with $k_1=0$. If $\frac{k}{\omega R^{3}}< \frac{2}{3\sqrt{3}}$, then (\ref{eq:actual_system}) has two limit cycles inside $D_{z}^{R}$, $\forall i\in\{0,1,\dots,n-1\}$.
        \begin{enumerate}
            \item The first limit cycle is described by the trajectory $L_{2}:=\{r^{\star}_{s_{2}} =r^{\star}_{s_{2}},\phi_{s_{2}}^{\star} =\psi_{s_{2}}^{\star}+\omega t,~\forall t\geq 0\}$.
            \item $L_{2}$ is asymptotically stable.
            \item The second limit cycle is described by the trajectory $L_{1}:=\{r^{\star}_{s_{1}} =r^{\star}_{s_{1}},\psi_{s_{1}}^{\star} =\psi_{s_{1}}^{\star}+\omega t,~\forall t\geq 0\}$.
            \item $L_{1}$ is unstable.
        \end{enumerate} 
        \label{thm:eql set theorem}
    \end{thm}
    \begin{proof}
    The proof follows directly from Lemmas \ref{lem_existance_of_eql_points}, \ref{special_lem_threeroots}, and Theorem \ref{thm:stability}. 
    \end{proof}
\subsection{Estimation of the Region of Attraction of $(r^{\star}_{s_{2}},\psi^{\star}_{s_{2}})$}\label{special_section3c}
Since the asymptotic stability of $(r^{\star}_{s_{2}},\psi^{\star}_{s_{2}})$ has already been established, we aim to estimate a maximal region of attraction (ROA). 
\subsubsection{ROA in $(u_s,v_s)$-frame}
First note that the asymptotic stable equilibrium set defined in \eqref{eq:R_E} specializes to $R_E=\{(r_{s_{2}}^{\star},\psi_{s_{2}}^{\star})\}$ and is not dependent on $\boldsymbol{\kappa}$. Hence, the stability set of $R_{E}$ defined earlier in \eqref{eq:stability set}, becomes:
\begin{align*}
    S = \{(r_s(0),\psi_s(0))|(r_s(t),\psi_s(t))\to R_E~ \text{as}~ t\to \infty\}.
\end{align*}
An estimate of $S$ (now equivalent to the conventional notion of ROA of $(r_{s_{2}}^{\star},\psi_{s_{2}}^{\star})$) can be computed using the optimization framed in Problem \ref{optimation_problem}. Note that in this case, since the solution of Problem \ref{optimation_problem}, i.e. $\Omega_s$, is not dependent on $\boldsymbol{\kappa}$, $\Omega_s$ is directly an estimate of $S$.
Two instances of $\Omega_{s}$ for the system \eqref{eq:spiral_ROA_shifted_equation_farthest} with $k_1=0$ are shown in Fig. \ref{ROA1_ROA2} in a two separate $\{u_s,v_s\}$ frames.

\subsubsection{ROA in $\{x_s,y_s\}$ frame}
Clearly, \eqref{eq:actual_system} is a time-varying system, and thus, the region of attraction of the stable limit cycle changes with time in the $\{x_s,y_s\}$ frame, while this set is invariant in the $\{u_s,v_s\}$ frame. It is evident that both the actual ROA and our ellipsoidal approximation, $\Omega_{s}$, are rotating around the origin in the $\{x_s,y_s\}$ frame at a constant angular velocity $\omega$. Two snapshots of this rotation of $\Omega_s$ are shown in Fig. \ref{ROA1_ROA2}. Clearly, $\Omega_s$ depends on the instantaneous orientation of the $\{u_s,v_s\}$ frame (denoted by $\theta$ in Fig. \ref{ROA1_ROA2}), which in turn depends on the relative position of the pursuer to the evaders. To address this $\{u_s,v_s\}$ dependence, we refine the notation of our estimated ROA to $\Omega_{s}(\theta)$ as shown in Fig. \ref{ROA1_ROA2}. Based on the above discussion, an appropriate notion of pursuer position independent ROA (PI-ROA) in the fixed $(x_s,y_s)$-frame in this case can be defined as follows.

Let the boundary of $D_{z}^{R}$ be represented as $\partial D_{z}^{R}$. Denote the solution of system (\ref{eq:evader_dynamics}) ($k_1=0$ and dropping the subscript $i$) corresponding to initial evader position $\{x_{s_{0}},y_{s_{0}}\}$ and initial pursuer position $\{P_{0}\}$, as $\phi_{s}(t;x_{s_{0}},y_{s_{0}},P_{0})$. Note that the $\phi_{s}(t;x_{s_{0}},y_{s_{0}},P_{0})$ will converge onto the stable limit cycle $L_{2}$.

\begin{defn}
    The PI-ROA of the limit set $L_{2}$ is defined as:
    $S^{\{x_s,y_s\}} =\{\{x_{s_{0}},y_{s_{0}}\}\in \mathbb{R}^{2}|\phi^{s}_{i}(t;x_{s_{0}},y_{s_{0}},P_{0})\to L_{2}\quad \text{as}~ t\to \infty,~\forall P_{0}\in \partial D_{z}^{R}\}$
\end{defn} 
Clearly, a simple approximation (though possibly conservative) of $S^{\{x_s,y_s\}}$ can be computed by the intersection:
\begin{align*}
    \hat{S}^{\{x_s,y_s\}}=\underset{\theta\in \{0,2\pi\}}\cap \Omega_{s}(\theta)
\end{align*}
An example of this intersection for $R=2, k = 1$ and $\omega=1$ is plotted in Fig. \ref{ROA_intersection_2}, where the pink circle encompasses $\hat{S}^{\{x_s,y_s\}}$. Evidently, any evader with initial position within $\hat{S}^{\{x_s,y_s\}}$ converges to $L_{2}$ as $t\to \infty$ regardless of the initial position of the pursuer on $D_{z}^{R}$.

\begin{figure}[ht]
\centering

\begin{subfigure}{0.49\linewidth}
    \centering
    \includegraphics[width=\linewidth]{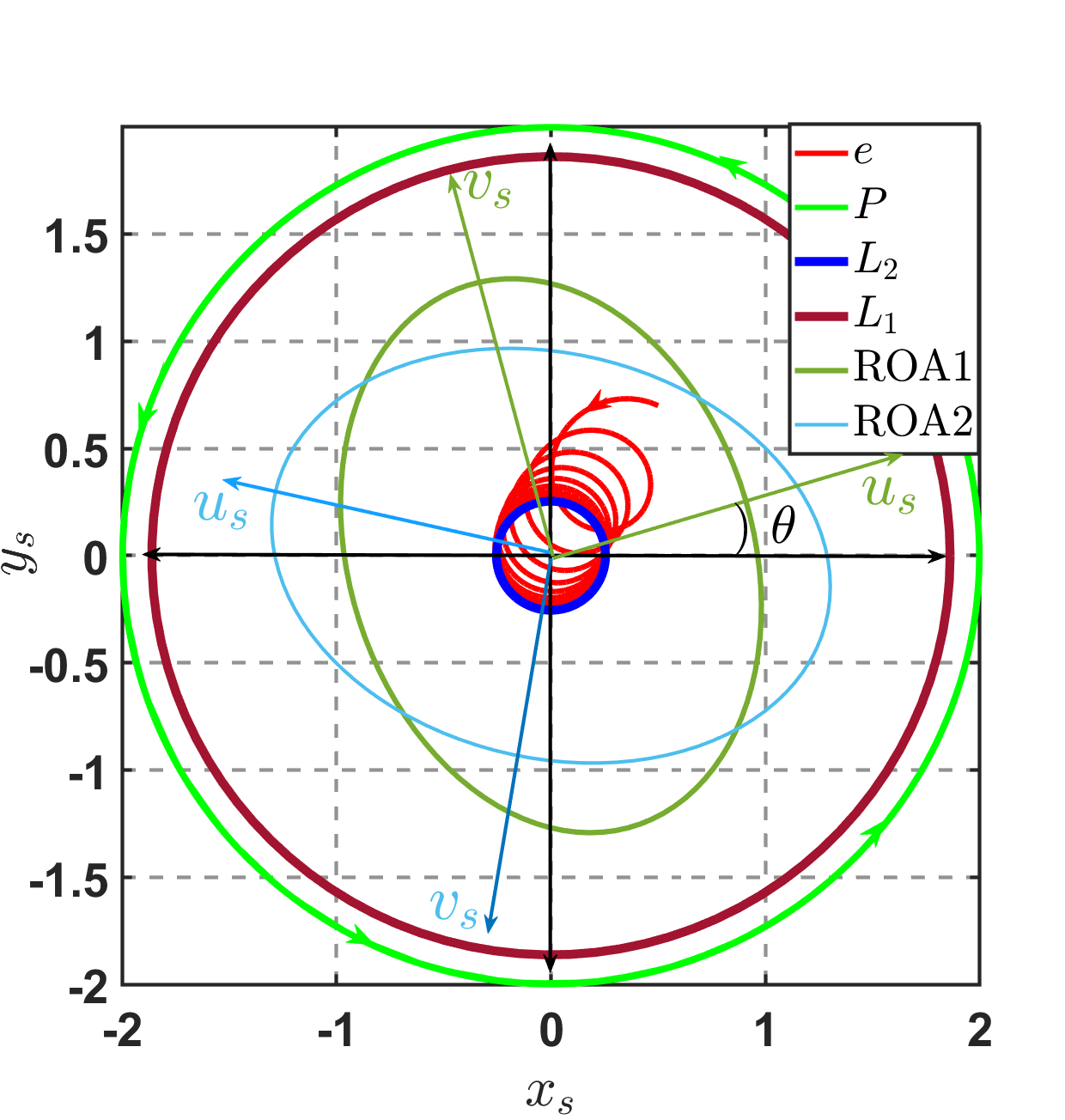}
    \caption{}
    \label{ROA1_ROA2}
\end{subfigure}%
\hfill
\begin{subfigure}{0.51\linewidth}
    \centering
    \includegraphics[width=\linewidth]{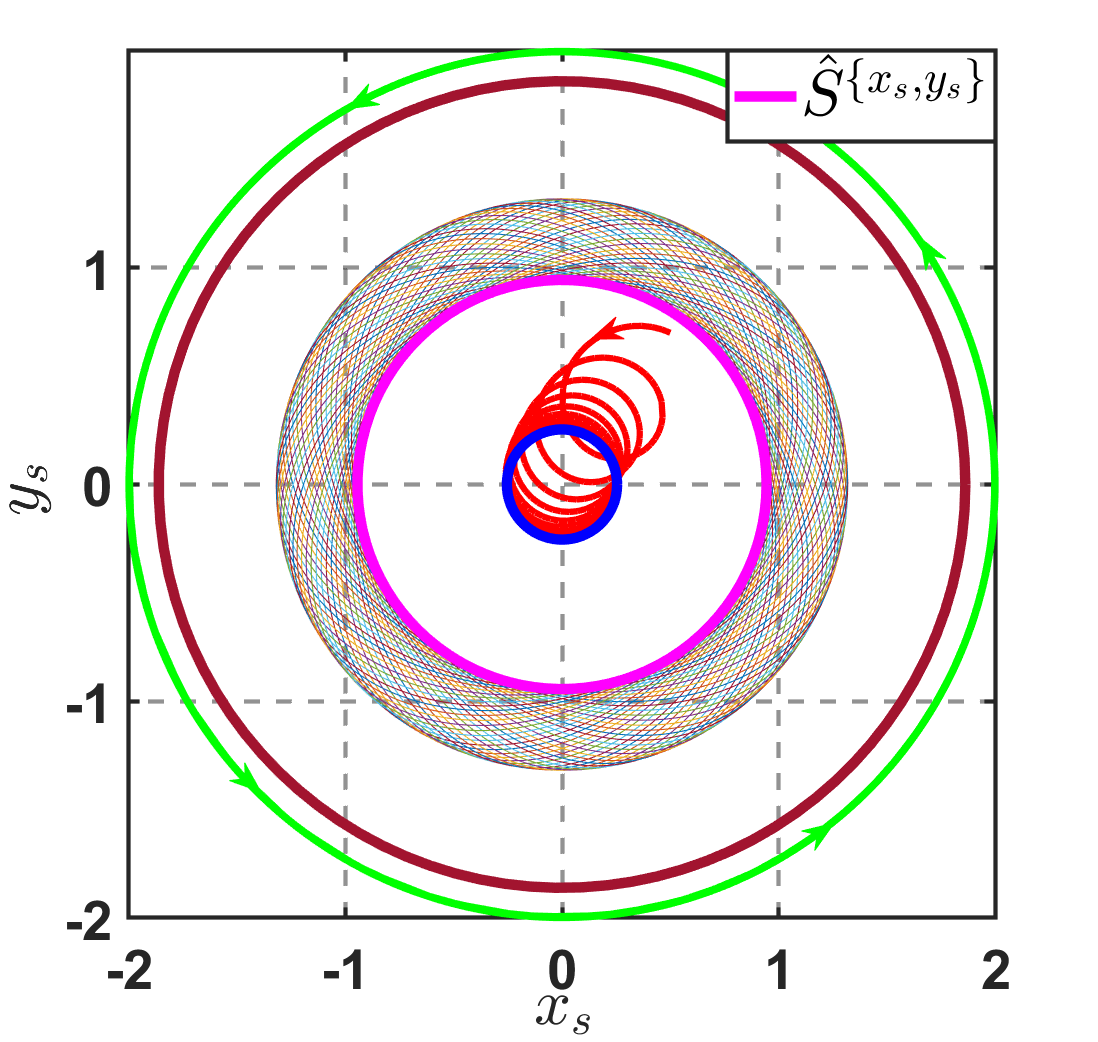}
    \caption{}
    \label{ROA_intersection_2}
\end{subfigure}

\caption{(a) Rotated ROA in $\{x,y\}$-frame, (b) Intersection of all ROA represents a circle in $\{x,y\}$-frame}
\label{ROA_interseaction}
\end{figure}

 We observe that the entire stability analysis from Section \ref{sec:special case} is independent of the specific position of each evader $e_{i}$ as long as $\{x_{e_{i}}(0),y_{e_{i}}(0)\}\in \hat{S}^{\{x_s,y_s\}}$. Clearly, the computation of the ROA is independent of the evader index and hence it is guaranteed that if $\{u^i_{s}(0),v^i_{s}(0)\}\in \Omega_s,~\forall i=\{0,1,\dots,n-1\}$ then each $u^i_{s} \to u^{i\star}_{s_{2}}=r^{i\star}_{s_{2}}\sin(\psi^{i\star}_{s_{2}})$ and $v^i_{s} \to v^{i\star}_{s_{2}}=r^{i\star}_{s_{2}}\cos(\psi^{i\star}_{s_{2}})$ as $t\to \infty$. In the $\{x_s,y_s\}$ frame, equivalently, Theorem \ref{thm:eql set theorem} holds for each evader, as does the computation of the PI-ROA, $\hat{S}^{\{x_s,y_s\}}$. 
\section{Simulation Results}\label{sec:simulation results}
This section illustrates the above results with typical convergent trajectories in various reference frames.
\subsection{Spiral pursuit: $k_1>0$}
\subsubsection{Single evader}
Figs. \ref{singh8} and \ref{singh9} depict the behavior
of an evader when $R=2$, $k=1$, $k_1=1$, and initial position $(x_{e_{0}},y_{e_{0}})=(0.7071,0.7071)$, with $\omega=2$ and $\omega=5$, respectively. We observe that the evader's trajectory forms a shifting spiral as it approaches the limit set, which is a circle with radius $r_{0}^{\star}$. Meanwhile, the pursuer converges to $R^{\star}$. Comparing these two figures, it is clear that $\omega$ is inversely related to the radius $r_{0}^{\star}$ of the limit set $L$, as expected from Lemma \ref{lem:r_limit_multiple}.

\begin{figure}[ht]
\centering

\begin{subfigure}{0.556\linewidth}
\hspace{-0.8cm}
    \centering
    \includegraphics[width=\linewidth]{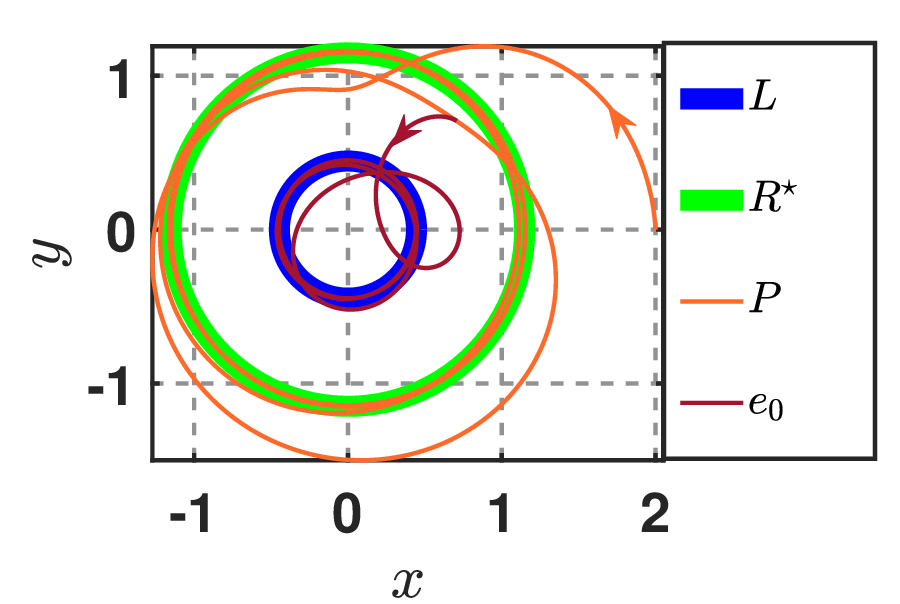}
    \caption{}
    \label{singh8}
\end{subfigure}%
\hspace{-0.5cm}
\begin{subfigure}{0.46\linewidth}
    \centering
    \includegraphics[width=\linewidth]{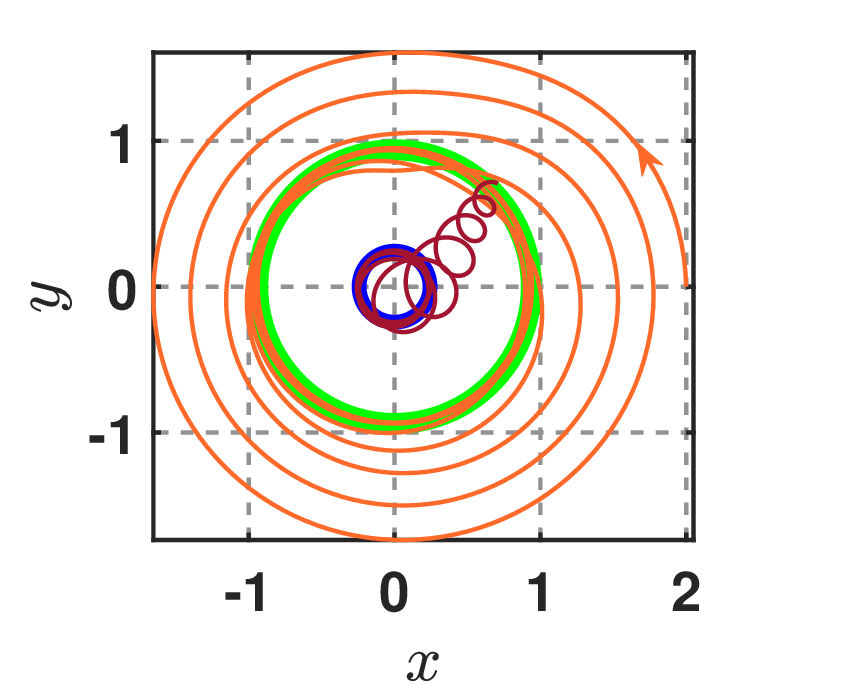}
    \caption{}
    \label{singh9}
\end{subfigure}

\caption{Evader and pursuer trajectories in $\{x,y\}$ coordinate frame for (a): $R=2, \omega=2, k=1, k_1=1$, resulting in $r_{0}^{\star}=0.4458$ and $R^{\star}=1.149$, (b): $R=2, \omega=5, k=1, k_1=1$, resulting in $r_{0}^{\star}=0.2434$ and $R^{\star}=0.9385$}
\label{single_pursuer_simulation_actual}
\end{figure}

Figs. \ref{singh10} and \ref{singh11} show the corresponding  evader trajectories in the $\{r,\psi\}/\{u,v\}$ frame, where the evader converges to the equilibrium point $(r_{0}^{\star},\psi_{0}^{\star})$. The effect of decreasing $\omega$ is evident from the reduced frequency of the spiral in Fig. \ref{single_pursuer_simulation_transferred}.
\begin{figure}[ht]
\centering
 \begin{subfigure}{0.5\linewidth}
     \includegraphics[width=\linewidth]{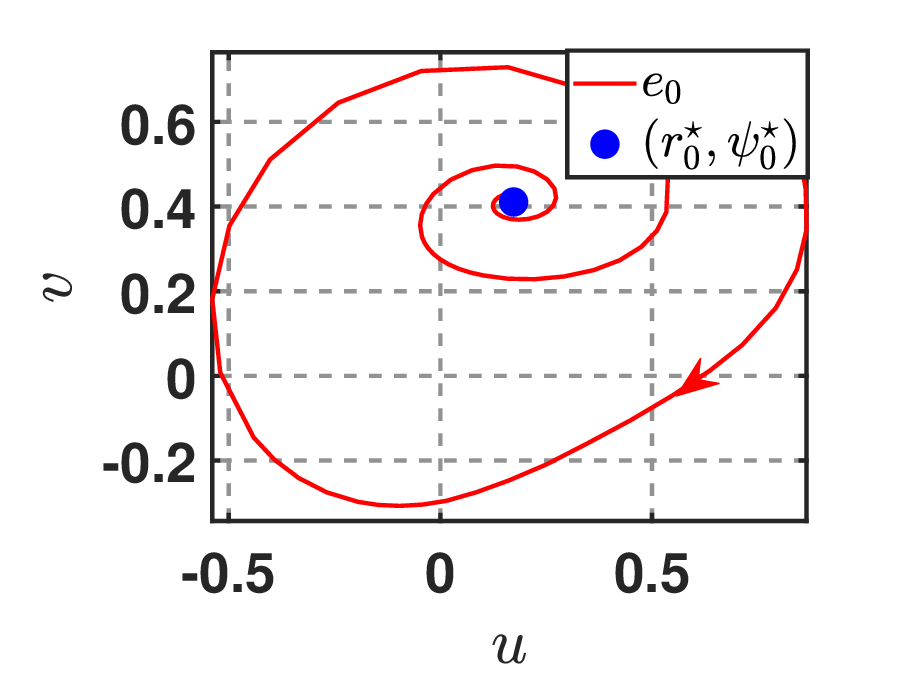}
     \caption{$\omega=2,r_{0}^{\star}=0.4458$}
     \label{singh10}
 \end{subfigure}
 \hfill
 \begin{subfigure}{0.48\linewidth}
 \centering
     \includegraphics[width=\linewidth]{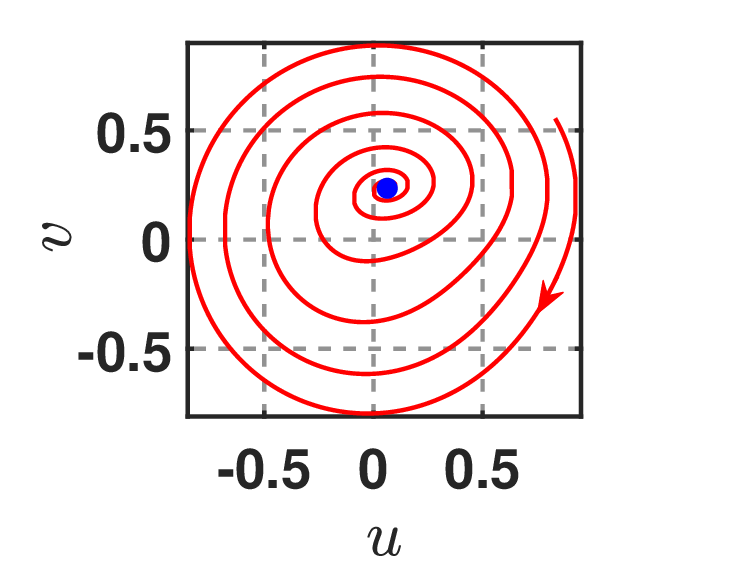}
     \caption{$\omega=5,r_{0}^{\star}=0.2434$}
     \label{singh11}
 \end{subfigure}
 \caption{Trajectories in $\{r,\psi\}$ frame for the two cases shown in Fig. \ref{single_pursuer_simulation_actual}}
 \label{single_pursuer_simulation_transferred}
\end{figure}
\subsubsection{Multiple evaders}
In the case of (multiple) two evaders, with $e_0$ as the outermost evader having initial position $(x_{e_{0}},y_{e_{0}})=(0.7071,0.7071)$, i.e., $\boldsymbol{\kappa}=1$, and another evader $e_i$, where $i\in\{1,\dots,n-1\}$, with initial position $(x_{e_{i}},y_{e_{i}})=(-0.3535,0.3535)$ as shown in Fig. \ref{singh12}, both evaders asymptotically converge onto the limit set $L$ in the $\{x,y\}/\{r,\phi\}$ frame.
Fig. \ref{singh13} illustrates the corresponding evader trajectories in the $\{u,v\}/\{r,\psi\}$ coordinate frame, where all evaders converge to the common equilibrium point $(r^{\star},\psi^{\star})=(0.4458,1.1728)$, regardless of the initial positions.
\begin{figure}[ht]
\centering
\hspace{-0.5cm}
 \begin{subfigure}{0.555\linewidth}
     
     \includegraphics[width=\linewidth]{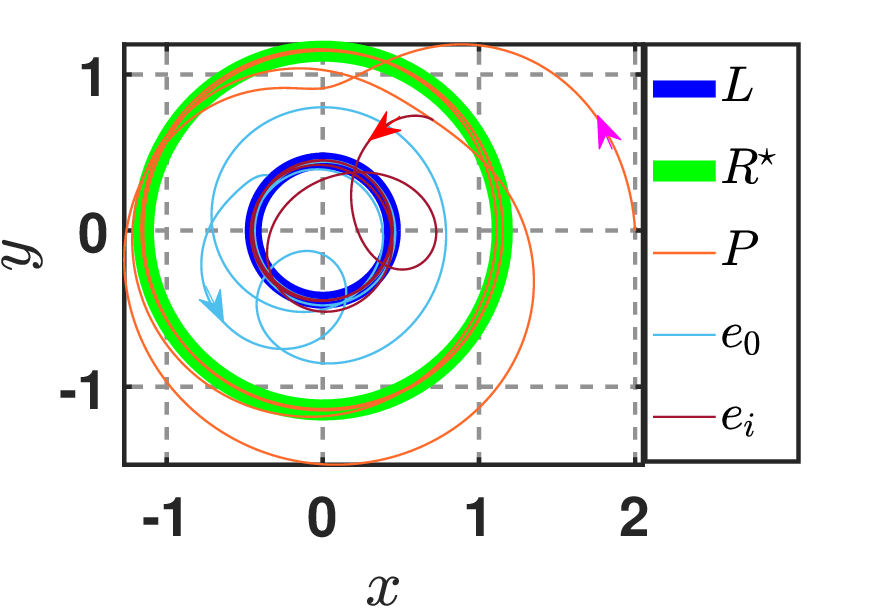}
     \caption{$\{x, y\}$ coordinate frame}
     \label{singh12}
 \end{subfigure}
 \hspace{-0.5cm}
 \begin{subfigure}{0.51\linewidth}
 \centering
     \includegraphics[width=\linewidth]{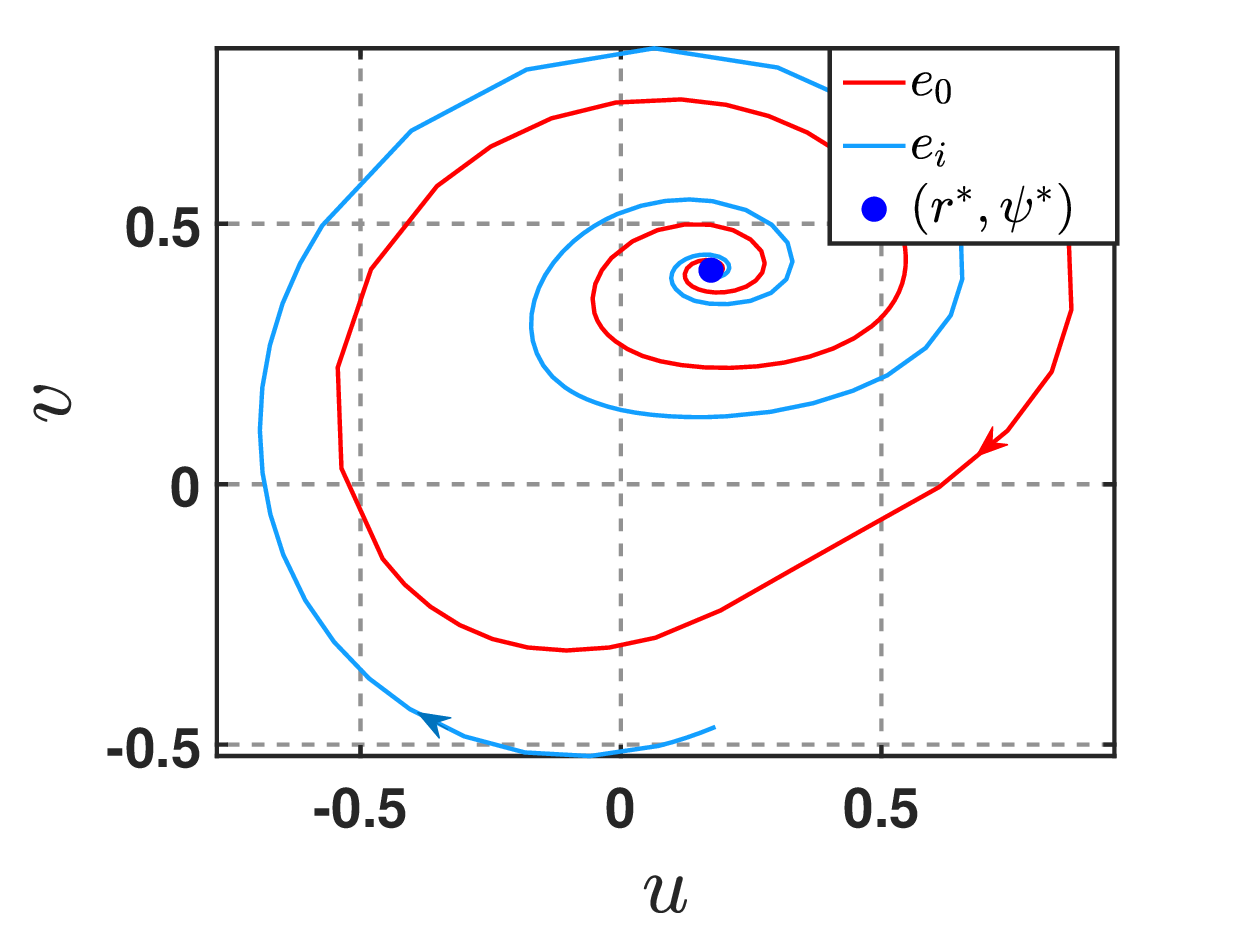}
     \caption{$\{r,\psi\}$ frame}
     \label{singh13}
 \end{subfigure}
 \caption{Evader and pursuer trajectory for $R = 2, \boldsymbol{\kappa}=1, \omega = 2, k = 1, k_1=1$}
 \label{multiple_pursuer_simulation_transferred}
\end{figure}
\subsection{Circular pursuit: $k_1=0$}
The multiple (three) evaders, shown in Fig. \ref{multiple_pursuer_simulation_actual}, asymptotically converge onto the limit set $L_{2}$ in the $\{x_s,y_s\}/\{r_s,\phi_s\}$ frame. 
\begin{figure}[ht]
\centering
 \begin{subfigure}{0.44\linewidth}
     \includegraphics[width=\linewidth]{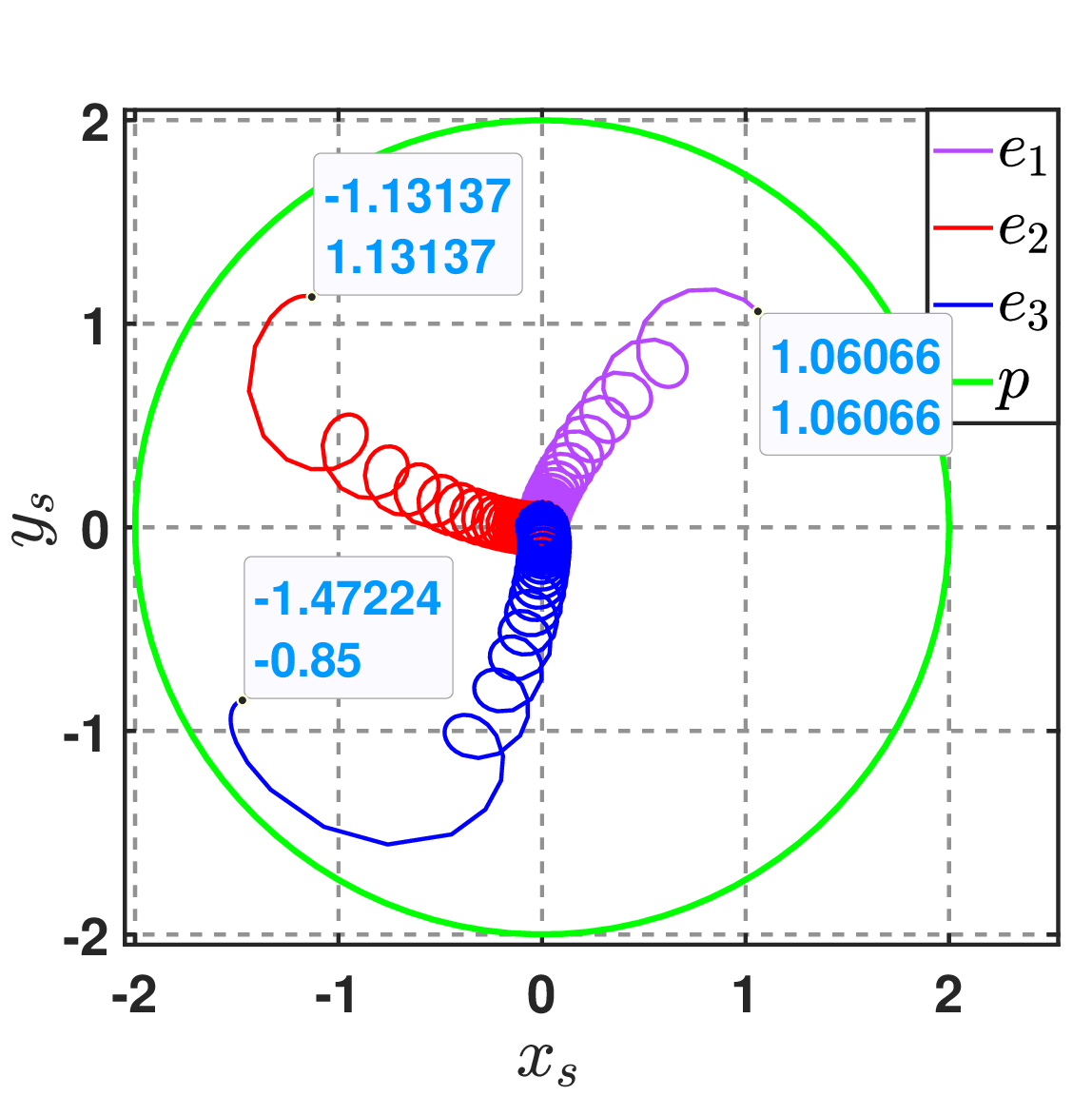}
     \caption{}
     \label{multiple_pursuer_simulation_actual}
 \end{subfigure}
 \hspace{0.2cm}
 \begin{subfigure}{0.51\linewidth}
 \centering
     \includegraphics[width=\linewidth]{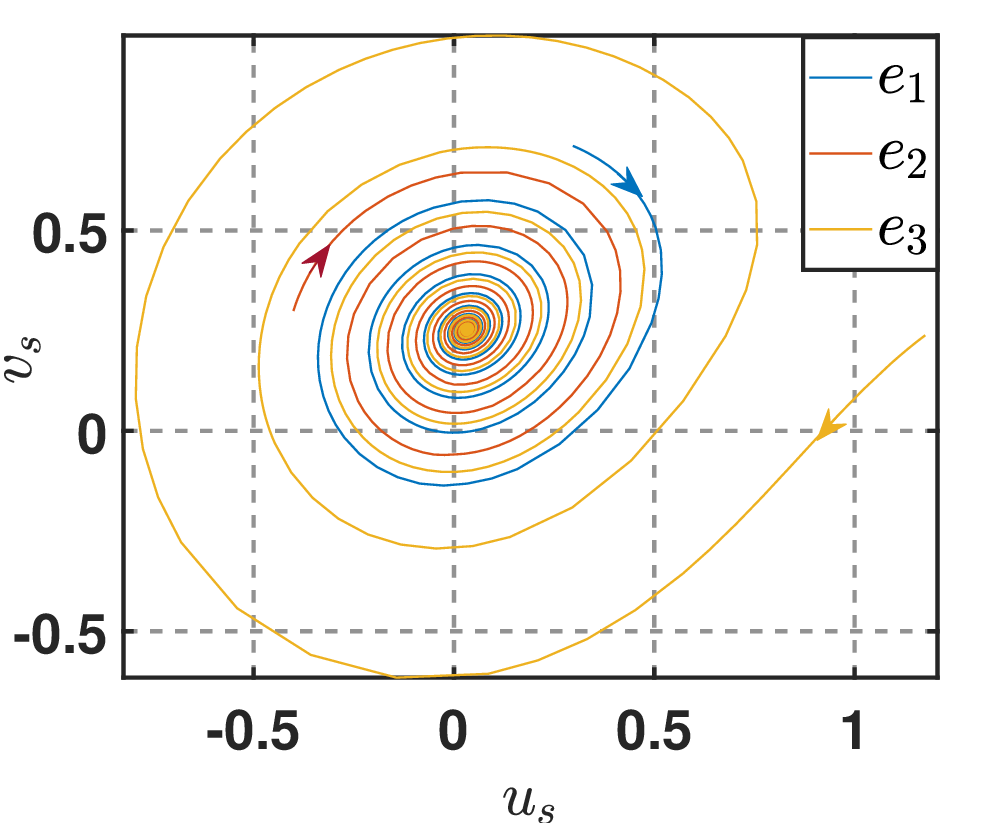}
     \caption{}
     \label{multiple_evader_simulation_transferred}
 \end{subfigure}
 \caption{Evader and pursuer trajectories in (a): $\{x_s,y_s\}$ coordinate frame, (b): $\{u_s,v_s\}$ coordinate frame for $R=2,\omega=2$, and $k=1$}
\end{figure}
Fig. \ref{multiple_evader_simulation_transferred} shows the corresponding evader trajectories in the $\{u_s,v_s\}/\{r_s,\psi_s\}$ coordinate frame, where all evaders converge to the common equilibrium point $(r^{\star}_{s_{2}},\psi_{s_{2}}^{\star})=(0.2541,1.4434)$, regardless of the initial positions.

\section{Conclusion}\label{conslusion}
In this paper, we have demonstrated that it is possible to herd multiple planar evaders into an arbitrarily small neighborhood of any target, provided they are located closer to the target than the pursuer at the start of the game. The resulting motions have been analyzed in detail, and various insights into the stability and region of attraction of the limiting periodic trajectories have been obtained. Variations of these pursuit laws should be capable of achieving more complex herding objectives and are currently under investigation.
\nocite{*}  
\bibliographystyle{IEEEtran}
\input{main.bbl}

\section*{Appendix A}
\input{./appendixA}
\section*{Appendix B}
\input{./appendixB}
\section*{Appendix C}
\input{./appendixC}
\end{document}

%% file: main.bbl

%% file: appendixA.tex
\textbf{Proof of Theorem \ref{thm:For-the-single}:} Linearizing \eqref{eq:transferred_dynamics} at equilibrium point $(r_{0}^{\star},\psi_{0}^{\star})$, the Jacobian matrix is: 
\begin{align}\label{eq:jacobian for farthest}
J_1 = \begin{bmatrix}
\frac{k\left(1 - Rk_{1}\cos(\psi_{0}^{\star})e^{k_{1}(r_{0}^{\star} - \boldsymbol{\kappa})}\right)}{\left(Re^{k_{1}(r_{0}^{\star} - \boldsymbol{\kappa})}\sin(\psi_{0}^{\star})\right)^{3}} & \frac{k}{\left(Re^{k_{1}(r_{0}^{\star} - \boldsymbol{\kappa})}\sin(\psi_{0}^{\star})\right)^{2}} \\
\frac{-k(2k_{1}r_{0}^{\star} + 1)}{\left(r_{0}^{\star}Re^{k_{1}(r_{0}^{\star} - \boldsymbol{\kappa})}\sin(\psi_{0}^{\star})\right)^{2}} & \frac{-2k\cos(\psi_{0}^{\star})}{r_{0}^{\star}R^{2}e^{2k_{1}(r_{0}^{\star} - \boldsymbol{\kappa})}\sin^{3}(\psi_{0}^{\star})}
\end{bmatrix}
\end{align}

Let $J_1(i,j)$ represents the $(i,j)th$ entry of Jacobian matrix $J_1$, where $(i,j)\in \{1,2\}\times\{1,2\}$.
Define $a=J_1(1,1), b=J_1(1,2), c=J_1(2,1)$ and $d=J_1(2,2)$. The eigenvalues of $J_1$ are:
\begin{align*}
\lambda_{1,2} & =\frac{a+d}{2}\pm\frac{\sqrt{(a+d)^{2}-4(ad-bc)}}{2}.\\
\end{align*} 
\textbf{Step 1: Analyzing $\frac{a+d}{2}$}
{\footnotesize
\begin{equation}
\begin{aligned}
    \frac{a+d}{2} &= \frac{k\left(1-Rk_{1}\cos(\psi_{0}^{\star})e^{k_{1}(r_{0}^{\star}-\boldsymbol{\kappa})}\right)}{\left(Re^{k_{1}(r_{0}^{\star}-\boldsymbol{\kappa})}\sin\left(\psi_{0}^{\star}\right)\right)^{3}} - \frac{-2k\cos(\psi_{0}^{\star})}{r_{0}^{\star}R^{2}e^{2k_{1}(r_{0}^{\star}-\boldsymbol{\kappa})}\sin^{3}\left(\psi_{0}^{\star}\right)}\\
    &= \frac{-k\cos{\psi_{0}^{\star}}-kRk_{1}(\cos^{2}(\psi_{0}^{\star}))e^{k_{1}(r_{0}^{\star}-\boldsymbol{\kappa})}}{R^3e^{3k_{1}(r_{0}^{\star}-\boldsymbol{\kappa})}\sin^3(\psi_{0}^{\star})\cos(\psi_{0}^{\star})}.
\end{aligned}
\end{equation}
}
Since $\cos(\psi_{0}^{\star})>0$ from \eqref{eq:r_eql} and $k,R,k_{1}, \boldsymbol{\kappa}>0$, we conclude that $\frac{a+d}{2}<0$.\\
\textbf{Step 2: Analyzing $(ad-bc)$}\\
    $ad-bc =$ 
\[
\frac{
    \left( 
    \begin{aligned}
        & -2k^2\cos^2(\psi_{0}^{\star}) + 2Rk_1k^2\cos^3(\psi_{0}^{\star})e^{k_1(r_{0}^{\star} - \boldsymbol{\kappa})} \\
        & \quad + 2r_{0}^{\star}k_1k^2\sin^2(\psi_{0}^{\star}) + k^2\sin^2(\psi_{0}^{\star})
    \end{aligned}
    \right)
}{
    R^6e^{6k_1(r_{0}^{\star}-\boldsymbol{\kappa})}\sin^6(\psi_{0}^{\star})\cos^2(\psi_{0}^{\star})
}.
\]
After further simplification, we get:\\
{\footnotesize
\begin{equation}
ad - bc = 
\frac{
    \left( 
    2k^2\cos(\psi_{0}^{\star}) (-\cos(\psi_{0}^{\star})+k_1Re^{k_1(r_{0}^{\star}-\boldsymbol{\kappa})})+ k^2\sin^2(\psi_{0}^{\star})
    \right)
}{
    R^6e^{6k_1(r_{0}^{\star}-\boldsymbol{\kappa})}\sin^6(\psi_{0}^{\star})\cos^2(\psi_{0}^{\star})
}
\label{eq:c4}
\end{equation}
}
The denominator $R^6e^{6k_1(r_{0}^{\star}-\boldsymbol{\kappa})}\sin^6(\psi_{0}^{\star})\cos^2(\psi_{0}^{\star})>0$. Also in the numerator, the terms $2k^2\cos(\psi_{0}^{\star})>0$ and $k^2\sin^2(\psi_{0}^{\star})>0$. Now consider the term $(-\cos(\psi_{0}^{\star})+k_1Re^{k_1(r_{0}^{\star}-\boldsymbol{\kappa})}))$. Using Theorem \ref{lem:unieq eql point}, we have:
\begin{align*}
    0>-\cos(\psi_{0}^{\star})>\frac{-1}{k_1Re^{1-\boldsymbol{\kappa}k_1}}
\end{align*}
Then:
{\footnotesize
\begin{align*}
    -\cos(\psi_{0}^{\star})+k_1Re^{k_1(r_{0}^{\star}-\boldsymbol{\kappa})})&>\frac{-1}{k_1Re^{(1-\boldsymbol{\kappa}k_1)}}+k_1Re^{k_1(r_{0}^{\star}-\boldsymbol{\kappa})},
\end{align*}
}
which simplifies to:
\begin{align}\label{eq:c6}
    -\cos(\psi_{0}^{\star})+k_1Re^{k_1(r_{0}^{\star}-\boldsymbol{\kappa})})>\frac{-1+k^2_1R^2e^{(1+k_1r_{0}^{\star}-2\boldsymbol{\kappa}k_1)}}{k_1Re^{(1-\boldsymbol{\kappa}k_1)}}.
\end{align}
On the right hand side of the inequality, denominator $k_1Re^{(1-\boldsymbol{\kappa}k_1)}>0$ and the numerator can be rewritten as 
\begin{align}\label{eq:c1}
    -1+k^2_1R^2e^{(1+k_1r_{0}^{\star}-2\boldsymbol{\kappa}k_1)} = -1+\underbrace{k^2_1R^2e^{-2\boldsymbol{\kappa}k_1}}_{T_1}\underbrace{e^{1+k_1r_{0}^{\star}}}_{T_2}
\end{align}
We have $k_1,r_{0}^{\star},R,\boldsymbol{\kappa}>0$ and from Lemma \ref{lem:one root existance}, we have 
\begin{align*}
    \boldsymbol{\kappa}<\frac{\ln{2k^2_1R^2}}{2k_1}\implies e^{-2\boldsymbol{\kappa}k_1}>\frac{1}{2k^2_1R^2}.
\end{align*}
Then from \eqref{eq:c1}
\begin{equation}\label{eq:c2}
\begin{aligned}
    &T_1=k^2_1R^2e^{-2\boldsymbol{\kappa}k_1}>\frac{k^2_1R^2}{2k^2_1R^2}=\frac{1}{2},\\
    & T_2 = e^{1+k_1r_{0}^{\star}}\geq 2.71.
\end{aligned}
\end{equation}
From \eqref{eq:c1} and \eqref{eq:c2}, we conclude:
\begin{align}\label{eq:c3}
-1+k^2_1R^2e^{-2\boldsymbol{\kappa}k_1}e^{1+k_1r_{0}^{\star}}>0.359.
\end{align}
From \eqref{eq:c4}, \eqref{eq:c6} and \eqref{eq:c3}, $ad-bc>0$ and $\frac{a+b}{2}<0$ which implies that $\mathfrak{Re}(\lambda_{1,2})<0$.

%% file: appendixB.tex
\textbf{Proof of Theorem \ref{thm:multiple_spiral}:} Linearizing \eqref{eq:transferred_dynamics} for the $j^{th}$ evader at the equilibrium point $(r^{\star},\psi^{\star})$, where $j\in\{1,2,\dots,n-1\}$, the Jacobian matrix is given by: 
\begin{align}\label{eq:jacobian for other}
J_2 = \begin{bmatrix}
\frac{k}{\left(Re^{k_{1}(r^{\star} - \boldsymbol{\kappa})}\sin(\psi^{\star})\right)^{3}} & \frac{k}{\left(Re^{k_{1}(r^{\star} - \boldsymbol{\kappa})}\sin(\psi^{\star})\right)^{2}} \\
\frac{-k}{\left(r^{\star}Re^{k_{1}(r^{\star} - \boldsymbol{\kappa})}\sin(\psi^{\star})\right)^{2}} & \frac{-2k}{r^{\star}R^{2}e^{2k_{1}(r^{\star} - \boldsymbol{\kappa})}\sin^{3}(\psi^{\star})}
\end{bmatrix}
\end{align}
Now, the overall Jacobian matrix for all the evaders, i.e., for $i = 0, 1, \dots, n-1$ using \eqref{eq:jacobian for farthest} and \eqref{eq:jacobian for other}, is given by:
\begin{align}
J = \begin{bmatrix}
J_1 & \bold{0} \\
B & J_2
\end{bmatrix},
\end{align}
where $\bold{0}\in\mathbb{R}^{2\times2}$ is the zero matrix, and $B\in\mathbb{R}^{2\times2}$ is not relevant for the current proof. Let $J_2(m,n)$ represents the $(m,n)th$ entry of Jacobian matrix $J_2$, where $(m,n)\in \{1,2\}\times\{1,2\}$.
Define $a=J_2(1,1), b=J_2(1,2), c=J_2(2,1)$ and $d=J_2(2,2)$. The eigenvalues of $J_2$ are:
\begin{align*}
\lambda_{1,2} & =\frac{a+d}{2}\pm\frac{\sqrt{(a+d)^{2}-4(ad-bc)}}{2}.\\
\end{align*} 
\textbf{Step 1: Analyzing $\frac{a+d}{2}$}
{\footnotesize
\begin{equation}
\begin{aligned}
    a+d &= \frac{k}{\left(Re^{k_{1}(r^{\star}-\boldsymbol{\kappa})}\sin\left(\psi^{\star}\right)\right)^{3}} - \frac{-2k}{r^{\star}R^{2}e^{2k_{1}(r^{\star}-\boldsymbol{\kappa})}\sin^{3}\left(\psi^{\star}\right)}\\
    &= \frac{k\cos{\psi^{\star}}-2k}{R^3e^{3k_{1}(r^{\star}-\boldsymbol{\kappa})}\sin^3(\psi^{\star})\cos(\psi^{\star})}.
\end{aligned}
\end{equation}
}
Since $1>\cos(\psi^{\star})>0$ from \eqref{eq:r_eql} and $k,R,k_{1}, \boldsymbol{\kappa}>0$, we conclude that $a+d<0$.\\
\textbf{Step 2: Analyzing $(ad-bc)$}\\    
\[
ad-bc = \frac{k(\sin^{2}(\psi^{\star})-2\cos^{2(\psi^{\star})})}{
    R^6e^{6k_1(r^{\star}-\boldsymbol{\kappa})}\sin^6(\psi^{\star})\cos^2(\psi^{\star})}.
\]
Since the denominator $R^6e^{6k_1(r^{\star}-\boldsymbol{\kappa})}\sin^6(\psi^{\star})\cos^2(\psi^{\star})>0$, we focus on numerator:
\begin{equation}\label{eq:one_three_cos_sqaure}
 \sin^{2}(\psi^{\star})-2\cos^{2}(\psi^{\star})=1-3\cos^{2}(\psi^{\star}).    
\end{equation}
Using \eqref{eq:cos_psi_bound}, we obtain:
\begin{align}\label{cos_sqaure}
    \cos^{2}(\psi^{\star})<\frac{1}{k_1^{2}R^{2}e^{2(1-\boldsymbol{\kappa}k_1)}}
\end{align}
Since: 
\begin{align}\label{eq:e^{2r_f_0k_1}}
    \boldsymbol{\kappa}<\frac{\log(2k_1^{2}R^{2}}{2k_1} \implies  e^{2k_1\boldsymbol{\kappa}}<2k_1^{2}R^{2}.
\end{align}
Substituting \eqref{eq:e^{2r_f_0k_1}} into \eqref{cos_sqaure}, we get:
\begin{equation}\label{eq:bound_cos_sqaure}
    \cos^{2}(\psi^{\star})<\frac{2}{e^2}<\frac{1}{3}.
\end{equation}
From the above steps, we deduce that $ad-bc>0$.
Since Step 1 shows $a+d<0$ and step 2 confirms $ad-bc>0$, it follows that $\mathfrak{Re}(\lambda_{1,2})<0$. Moreover, as shown in Appendix A, the eigenvalues of $J_1$ also have negative real parts. Hence, the equilibrium point is locally asymptotically stable.

%% file: appendixC.tex
\begin{lem}\label{lem:simplified_rs1}
    For fixed $k$ and $\omega$, the value of $r_{s_{1}}^{\star}$ is given by:    \begin{equation*}
       r_{s_{1}}^{\star} = \frac{4R^{2}}{3}\cos^{2}\left(\frac{\theta_{1}}{3}\right),
    \end{equation*}
\end{lem}
where $\theta_1$ is defined as:
\begin{equation*}
    \theta_{1}=\pi-\tan^{-1}\left(\frac{2\omega\sqrt{-\frac{k^{2}}{4\omega^{2}}+\frac{R^{6}}{27}}}{k}\right).
\end{equation*}
\begin{proof}
    Recall from (\ref{special_threeroots}) that
\begin{equation*}
    r_{s_{1}}^{\star} =\left(-\sqrt{\frac{k^{2}}{4\omega^{2}}-\frac{R^{6}}{27}}-\frac{k}{2\omega}\right)^{\frac{1}{3}} + \left(\sqrt{\frac{k^{2}}{4\omega^{2}}-\frac{R^{6}}{27}}-\frac{k}{2\omega}\right)^{\frac{1}{3}}
\end{equation*}
and also recall that
 $\frac{k}{\omega R^3}<\frac{2}{3\sqrt{3}} \Rightarrow \frac{k^{2}}{4\omega^{2}}-\frac{R^{6}}{27}<0$.
Let 
\begin{align}\label{eq:c}
    c:=\sqrt{-\frac{k^{2}}{4\omega^{2}}+\frac{R^{6}}{27}},
\end{align}
where $c>0$. Then,
\begin{align*}
r_{s_{1}}^{*^2} &= \left(\left(cj-\frac{k}{2\omega}\right)^{\frac{1}{3}} + \left(-cj-\frac{k}{2\omega}\right)^{\frac{1}{3}}\right)^{2}\\
&= \left(\left(\sqrt{c^2+\frac{k^2}{4 \omega^2}}\right)^{\frac{1}{3}} e^{\frac{\theta_{1}j}{3}}+\left(\sqrt{c^2+\frac{k^2}{4 \omega^2}}\right)^{\frac{1}{3}} e^{\frac{\theta_{2}j}{3}}\right)^2,
\end{align*}
where 
\begin{align}\label{eq:theta1}
    \theta_{1}=\pi-\tan^{-1}(\frac{2\omega c}{k})
\end{align}
and $\theta_{2}=-\pi+\tan^{-1}(\frac{2\omega c}{k})$. Note that $\theta_{1}=-\theta_{2}$. Substituting $\theta_{2}=-\theta_{1}$ in the above equation, we get 

\begin{align}\label{special_r1_star_square}
r_{s_{1}}^{*2} &= \left(c^2+\frac{k^2}{4 \omega^2}\right)^\frac{1}{3}\left(2 \cos\left(\frac{\theta_{1}}{3}\right)\right)^{2}\\
    &= 4\left(c^2+\frac{k^2}{4 \omega^2}\right)^\frac{1}{3}\cos^{2}\left(\frac{\theta_{1}}{3}\right).    
\end{align}
Next observe that the term $\frac{2c\omega}{k}$  appears in the expression for $\theta_{1}$. Using \eqref{eq:c} in $\frac{2\omega c}{k}$, we get
\begin{align}\label{2cw/k_equation}
    \frac{2c\omega}{k}&=\frac{2\omega\sqrt{-\frac{k^{2}}{4\omega^{2}}+\frac{R^{6}}{27}}}{k} = 2\sqrt{\frac{R^{6}\omega^{2}}{27k^{2}}-\frac{1}{4}} \nonumber\\
    &= 2\sqrt{\frac{1}{27\left(\frac{k}{\omega R^{3}}\right)^{2}}-\frac{1}{4}}.
\end{align}
Applying \eqref{2cw/k_equation} in \eqref{eq:theta1}, and \eqref{eq:c} in \eqref{special_r1_star_square}, we obtain the desired result.

\end{proof}

\begin{lem}\label{monotonically}
    For fixed $k$, $z:=4\omega^{2}(r_{s_{1}}^{*^2}-R^{2})^{3}$ is a monotonically increasing function of $(\omega R^{3})$, where $\frac{3\sqrt{3}k}{2}<\omega R^{3}<\infty$.
\end{lem}
\begin{proof}
    Substituting $r_{s_{1}}^{\star}$ from Lemma \ref{lem:simplified_rs1}, we obtain:
    \small
    \begin{align}\label{z_equation}
        z &= 4\omega^{2}\left(\frac{4R^{2}}{3}\cos^{2}\left(\frac{\theta_{1}}{3}\right)-R^{2}\right)^{3} \nonumber\\
         &=\underbrace{ 4\omega^{2}R^{6}}_{f}  \underbrace{\left(\frac{4}{3}\cos^{2}\left(\frac{\pi}{3}-\frac{1}{3}\tan^{-1}\left(2\sqrt{\frac{1}{27\left(\frac{k}{\omega R^{3}}\right)^{2}}-\frac{1}{4}}\right)\right)-1\right)^{3}}_{g}\\ \nonumber
         &=: f(\omega R^{3}) g(\omega R^{3}) \quad \text{(say)}.
    \end{align}
    \normalsize
    Now it is easy to verify that $g(\omega R^{3})$ is an monotonically increasing function of $\omega R^{3}$ for $\frac{3\sqrt{3}k}{2}<\omega R^{3}<\infty$, while $f(\omega R^{3})$ is quadratic in $\omega R^{3}$. Hence the product $z=fg$ is an monotonically increasing function of $\omega R^{3}$.
\end{proof}

\begin{lem}\label{upper_bound_on_wR^3}
    For fixed $k$, $z\to 0$ as $\omega R^{3}\to \infty$
\end{lem}

\begin{proof}
For (48) above, it can be shown using L'Hospitals rule that \(z \to 0\) as \(\omega R^3 \to \infty\)\
\end{proof}

\begin{lem}\label{lower_bound_on_wR^3}
    For fixed $k$, $z\to -8k^{2}$ as $\omega R^{3}\to \frac{3\sqrt{3}k}{2}$
\end{lem}

\begin{proof}
    Using $\omega R^{3}=\frac{3\sqrt{3}k}{2}$ in \eqref{z_equation},
  \begin{align*}
      z &= 4\omega^{2}R^{6}\left(\frac{4}{3}\cos^{2}\left(\frac{\pi}{3}-\frac{1}{3}\tan^{-1}\left(0\right)\right)-1\right)^{3}\\
      &= -\frac{32}{27} \omega^{2} R^{6}\\
      &= -\frac{32}{27} \left(\frac{27k^{2}}{4}\right) ~~(\text{using}~~ \frac{k}{\omega R^{3}}=\frac{2}{3\sqrt{3}})\\
      &= -8k^{2}.
  \end{align*}
\end{proof}

\textbf{Proof of claim 2 of Theorem \ref{thm:stability}:}
To prove claim 2, 
Now, recall the expression for the eigenvalues:
\begin{align*}
    \lambda_{1,2}=\frac{1}{2(R^{2}-r_{s_{1}}^{*{2}})^{\frac{3}{2}}}\left(-k\pm\sqrt{9k^{2}-4\omega^{2}(R^{2}-r_{s_{1}}^{*2})^{3}}\right)
\end{align*}
and let
\begin{align}
    y &= 9k^{2}-4\omega^{2}(R^{2}-r_{s_{1}}^{*^2})^{3} \nonumber\\
      &= 9k^{2} + 4\omega^{2}\left(r_{s_{1}}^{*^2}-R^{2}\right)^{3} =: 9k^{2}+z.
\end{align}
Since lemma \ref{lem:simplified_rs1}, \ref{monotonically}, \ref{upper_bound_on_wR^3}, and \ref{lower_bound_on_wR^3} are valid for any $k$, $y=9k^{2}+z$ is also monotonically increasing with $\omega R^{3}$, with the bounds $9k^{2}-8k^{2}\leq y\leq 9k^{2} -0\Leftrightarrow k^{2}\leq y\leq 9k^{2}$. This implies both the roots $\lambda_{1,2}=\frac{1}{2(R^{2}-r_{s}^{*^{2}})^{\frac{3}{2}}}\left(-k\pm\sqrt{9k^{2}-4\omega^{2}(R^{2}-r_{s}^{*2})^{3}}\right)$ are always real and of opposite signs.

%% file: main.bbl
\begin{thebibliography}{10}
\providecommand{\url}[1]{#1}
\csname url@samestyle\endcsname
\providecommand{\newblock}{\relax}
\providecommand{\bibinfo}[2]{#2}
\providecommand{\BIBentrySTDinterwordspacing}{\spaceskip=0pt\relax}
\providecommand{\BIBentryALTinterwordstretchfactor}{4}
\providecommand{\BIBentryALTinterwordspacing}{\spaceskip=\fontdimen2\font plus
\BIBentryALTinterwordstretchfactor\fontdimen3\font minus \fontdimen4\font\relax}
\providecommand{\BIBforeignlanguage}[2]{{%
\expandafter\ifx\csname l@#1\endcsname\relax
\typeout{** WARNING: IEEEtran.bst: No hyphenation pattern has been}%
\typeout{** loaded for the language `#1'. Using the pattern for}%
\typeout{** the default language instead.}%
\else
\language=\csname l@#1\endcsname
\fi
#2}}
\providecommand{\BIBdecl}{\relax}
\BIBdecl

\bibitem{vaughan2000experiments}
R.~Vaughan, N.~Sumpter, J.~Henderson, A.~Frost, and S.~Cameron, ``Experiments in automatic flock control,'' \emph{Robotics and autonomous systems}, vol.~31, no. 1-2, pp. 109--117, 2000.

\bibitem{wood2007evolving}
A.~J. Wood and G.~J. Ackland, ``Evolving the selfish herd: emergence of distinct aggregating strategies in an individual-based model,'' \emph{Proceedings of the Royal Society B: Biological Sciences}, vol. 274, no. 1618, pp. 1637--1642, 2007.

\bibitem{bbc_article}
``“helicopter cowboys of australia's outback”, https://www.bbc.com/news/world-asia-pacific-12408888.''

\bibitem{van2023steering}
S.~Van~Havermaet, P.~Simoens, T.~Landgraf, and Y.~Khaluf, ``Steering herds away from dangers in dynamic environments,'' \emph{Royal Society Open Science}, vol.~10, no.~5, p. 230015, 2023.

\bibitem{hughes2003flow}
R.~L. Hughes, ``The flow of human crowds,'' \emph{Annual review of fluid mechanics}, vol.~35, no.~1, pp. 169--182, 2003.

\bibitem{chipade2021aerial}
V.~S. Chipade, V.~S.~A. Marella, and D.~Panagou, ``Aerial swarm defense by stringnet herding: Theory and experiments,'' \emph{Frontiers in Robotics and AI}, vol.~8, p. 640446, 2021.

\bibitem{long2020comprehensive}
N.~K. Long, K.~Sammut, D.~Sgarioto, M.~Garratt, and H.~A. Abbass, ``A comprehensive review of shepherding as a bio-inspired swarm-robotics guidance approach,'' \emph{IEEE Transactions on Emerging Topics in Computational Intelligence}, vol.~4, no.~4, pp. 523--537, 2020.

\bibitem{lien2004shepherding}
J.-M. Lien, O.~B. Bayazit, R.~T. Sowell, S.~Rodriguez, and N.~M. Amato, ``Shepherding behaviors,'' in \emph{IEEE International Conference on Robotics and Automation, 2004. Proceedings. ICRA'04. 2004}, vol.~4.\hskip 1em plus 0.5em minus 0.4em\relax IEEE, 2004, pp. 4159--4164.

\bibitem{miki2006effective}
T.~Miki and T.~Nakamura, ``An effective simple shepherding algorithm suitable for implementation to a multi-mmobile robot system,'' in \emph{First International Conference on Innovative Computing, Information and Control-Volume I (ICICIC'06)}, vol.~3.\hskip 1em plus 0.5em minus 0.4em\relax IEEE, 2006, pp. 161--165.

\bibitem{strombom2014solving}
D.~Str{\"o}mbom, R.~P. Mann, A.~M. Wilson, S.~Hailes, A.~J. Morton, D.~J. Sumpter, and A.~J. King, ``Solving the shepherding problem: heuristics for herding autonomous, interacting agents,'' \emph{Journal of the royal society interface}, vol.~11, no. 100, p. 20140719, 2014.

\bibitem{hamilton1971geometry}
W.~D. Hamilton, ``Geometry for the selfish herd,'' \emph{Journal of theoretical Biology}, vol.~31, no.~2, pp. 295--311, 1971.

\bibitem{werner1993evolution}
G.~M. Werner and M.~G. Dyer, ``Evolution of herding behavior in artificial animals,'' \emph{from animals to animats}, vol.~2, pp. 393--399, 1993.

\bibitem{scott2013pursuit}
W.~Scott and N.~E. Leonard, ``Pursuit, herding and evasion: A three-agent model of caribou predation,'' in \emph{2013 American Control Conference}.\hskip 1em plus 0.5em minus 0.4em\relax IEEE, 2013, pp. 2978--2983.

\bibitem{shedied2002optimal}
S.~A. Shedied, ``Optimal control for a two-player dynamic pursuit evasion game: The herding problem,'' Ph.D. dissertation, Virginia Polytechnic Institute and State University, 2002.

\bibitem{kachroo2001dynamic}
P.~Kachroo, S.~A. Shedied, J.~S. Bay, and H.~Vanlandingham, ``Dynamic programming solution for a class of pursuit evasion problems: the herding problem,'' \emph{IEEE Transactions on Systems, Man, and Cybernetics, Part C (Applications and Reviews)}, vol.~31, no.~1, pp. 35--41, 2001.

\bibitem{khalafi2011capture}
A.~D. Khalafi and M.~R. Toroghi, ``Capture zone in the herding pursuit evasion games,'' \emph{Appl. Math. Sci}, vol.~5, no.~39, pp. 1935--1945, 2011.

\bibitem{lien2005shepherding}
J.-M. Lien, S.~Rodriguez, J.~Malric, and N.~M. Amato, ``Shepherding behaviors with multiple shepherds,'' in \emph{Proceedings of the 2005 IEEE International Conference on Robotics and Automation}.\hskip 1em plus 0.5em minus 0.4em\relax IEEE, 2005, pp. 3402--3407.

\bibitem{lu2010cooperative}
Z.~Lu, ``Cooperative optimal path planning for herding problems,'' Ph.D. dissertation, Texas A \& M University, 2006.

\bibitem{pierson2018controlling}
A.~Pierson and M.~Schwager, ``Controlling noncooperative herds with robotic herders,'' \emph{IEEE Transactions on Robotics}, vol.~34, no.~2, pp. 517--525, 2018.

\bibitem{gadre2001learning}
A.~S. Gadre, ``Learning strategies in multi-agent systems-applications to the herding problem,'' Ph.D. dissertation, Virginia Tech, 2001.

\bibitem{bacon2012swarm}
M.~Bacon and N.~Olgac, ``Swarm herding using a region holding sliding mode controller,'' \emph{Journal of Vibration and Control}, vol.~18, no.~7, pp. 1056--1066, 2012.

\bibitem{licitra2017singleadaptive}
R.~A. Licitra, Z.~I. Bell, E.~A. Doucette, and W.~E. Dixon, ``Single agent indirect herding of multiple targets: A switched adaptive control approach,'' \emph{IEEE Control Systems Letters}, vol.~2, no.~1, pp. 127--132, 2017.

\bibitem{licitra2019single}
R.~A. Licitra, Z.~I. Bell, and W.~E. Dixon, ``Single-agent indirect herding of multiple targets with uncertain dynamics,'' \emph{IEEE Transactions on Robotics}, vol.~35, no.~4, pp. 847--860, 2019.

\bibitem{licitra2017singleswitched}
R.~A. Licitra, Z.~D. Hutcheson, E.~A. Doucette, and W.~E. Dixon, ``Single agent herding of n-agents: A switched systems approach,'' \emph{IFAC-PapersOnLine}, vol.~50, no.~1, pp. 14\,374--14\,379, 2017.

\bibitem{rishabh2024}
R.~K. Singh and D.~Chakraborty, ``Planar herding of multiple evaders by a single pursuer,'' in \emph{2024 IEEE 63rd Conference on Decision and Control (CDC)}, 2024, pp. 7375--7380.

\bibitem{curtiss1918recent}
D.~Curtiss, ``Recent extentions of descartes' rule of signs,'' \emph{Annals of Mathematics}, vol.~19, no.~4, pp. 251--278, 1918.

\bibitem{hirsch1974differential}
M.~W. Hirsch, R.~L. Devaney, and S.~Smale, \emph{Differential equations, dynamical systems, and linear algebra}.\hskip 1em plus 0.5em minus 0.4em\relax Academic press, 1974, vol.~60.

\bibitem{davison1971computational}
E.~Davison and E.~Kurak, ``A computational method for determining quadratic lyapunov functions for non-linear systems,'' \emph{Automatica}, vol.~7, no.~5, pp. 627--636, 1971.

\bibitem{fmincon_matlab}
``https://in.mathworks.com/help/optim/ug/fmincon.html.''

\bibitem{abramowitz1968handbook}
M.~Abramowitz and I.~A. Stegun, \emph{Handbook of mathematical functions with formulas, graphs, and mathematical tables}.\hskip 1em plus 0.5em minus 0.4em\relax US Government printing office, 1968, vol.~55.
\BIBentrySTDinterwordspacing

\end{thebibliography}
